\documentclass[letter,11pt]{article}
\usepackage{amsmath} %equation*,pmatrix,cases
\usepackage{amssymb, dsfont, mathrsfs}  % \mathbb
\usepackage{amsthm}
\usepackage{graphicx}
\usepackage{algorithm}
\usepackage{algorithmic}

\usepackage{times}
\allowdisplaybreaks[3]
\newcommand{\SIDE}{{\rm SIDE}}
\newcommand{\Myquit}{q_{{\rm Mid}}}
\newtheorem{theorem}{Theorem}[section]
\newtheorem{lemma}[theorem]{Lemma}

\newtheorem{proposition}[theorem]{Proposition}

\newtheorem{problem}{Problem}[section]

\title{Can Walker Localize The Middle Point of A Line-segment?}
\author{
Akihiro Monde\footnotemark[1] \and
Yukiko Yamauchi\footnotemark[1] \and
Shuji Kijima\footnotemark[1] \footnotemark[2] \and
Masafumi Yamashita\footnotemark[1] 
}
\begin{document}
\maketitle
\renewcommand\thefootnote{\fnsymbol{footnote}}
\footnotetext[1]{Graduate School of Information Science and Electronic Engineering, Kyushu University}
\footnotetext[2]{Corresponding: kijima@inf.kyushu-u.ac.jp}
\renewcommand\thefootnote{\arabic{footnote}}

\begin{abstract}
 This paper poses a question about a simple localization problem. 
%   which is arisen from {\em self-stabilizing} location problems 
%   by {\em autonomous mobile robots} with {\em limited visibility}. 
%    that is a widely interested abstract model in distributed computing. 
 The question is if 
   an {\em oblivious} walker on a line-segment 
   can localize the middle point of the line-segment in {\em finite} steps 
   observing the direction (i.e., Left or Right) and the distance to the nearest end point. 
% This problem is also akin to (a continuous version of) binary search, and 
%   could be closely related to computable real functions. 
 This problem is arisen from {\em self-stabilizing} location problems 
   by {\em autonomous mobile robots} with {\em limited visibility}, 
    that is a widely interested abstract model in distributed computing. 
 Contrary to appearances, 
   it is far from trivial if this simple problem is solvable or not, and unsettled yet. 
 This paper is concerned with three variants of the problem with a minimal relaxation, and 
   presents self-stabilizing algorithms for them. 
 We also show an easy impossibility theorem 
   for bilaterally symmetric algorithms. 

%%%%%%%%%%
%\keywords{
\end{abstract}

\section{Introduction}\label{sec:intro}
%%%%%%%%%%%%%%%%%%%%%%%%%%%%%%%%%%%
%\subsection{Problem: localization of the midpoint}\label{sec:intro-prob}
%%%%%%%%%%%%
% Ani Walker (Anarchy Ropewalker, Yaji) is a perfect balancer. 
 Ani Walker is a perfect balancer. 
 He is on a long rope in the sky, and  
  he wants to get to the exactly middle point of the rope; 
  from the point he can jump down to a safety net below. 
 Two instructors stand on the both ends, 
   Master Light ($L$) is on the left-end and 
   Master Dark ($R$) is on the right-end. 
 Walker asks them a query ``where am I?'' 
 They answer ``you are standing at the midpoint'' 
   if Walker is exactly at the midpoint, 
  otherwise the master in the nearest side 
   answers ``you are located in my side distance $d$ from me.''  
 Can Walker localize the midpoint of the rope? 

%%%%%%%
 Let the rope be denoted by a closed real interval $[-D,D] \subseteq \mathbb{R}$ with a positive real $D$. 
 For a location query by Walker at $x \in [-D,D]$, 
  our oracle returns the halting state $\Myquit$ if $x=0$, 
  otherwise, 
   it returns the nearest-end $\SIDE(x) \in \{L,R\}$ and 
   the distance $d(x)$ between $x$ and $\SIDE(x)$, 
   i.e., 
    Walker observes $\SIDE(x)=R$ and $d(x) = D-x$ if $x>0$, and 
    he observes $\SIDE(x)=L$ and $d(x) = D+x$ if $x<0$. 
 Of course, he does not know neither $D$ nor $x$. 
 Unless Walker locates at the point $x=0$, 
  he deterministically decides (and moves to) the next position $x'$, 
   using $\SIDE(x)$ and $d(x)$. 
 The question is 
   if there is a way for Walker to reach at $x_t = 0$ in a finite steps $t$ 
   starting from an arbitrary $x_0 \in [-D,D]$.

%%%%%%%%%%%%%%
 The answer is YES, 
  if Walker uses the history of the positions and the motions to the current position: 
 Firstly he visits the left-end, 
   here $x_0$ denotes the left-end. 
 Then, 
  he asks a query every one step as long as he is in the side $L$; 
  i.e., let $x_{t+1}=x_t+1$ for $t=0,1,2,\ldots$ if $\SIDE(x_t)=L$. 
 Once he observes $\SIDE(x_n) = R$ and $d(x_n)=d$, 
  he knows that the length of the path, that is $2D=n+d$. 
 Thus, the midpoint is placed at $x_n-(\frac{n+d}{2}-d)$, 
  thus he should just move to left distance $\frac{n-d}{2}$ from $x_n$.

%%%%%%%%%Problem%%%%%%%%%%%%
\begin{figure}[tbp]
\begin{center}
\includegraphics[width=90mm]{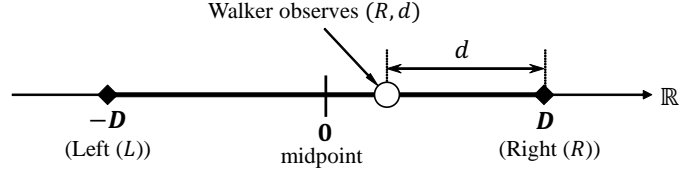}
\caption{
 A sketch of the problem. 
 Walker is on a line segment. 
 He observe the direction and the distance to the nearest end-point $(R,d)$, and  
   moves to the next place based on $(R,d)$. 
% The goal is to localize eventually the middle point of the line-segment 
% (see Problem~\ref{prob:original} for more detail). 
 Is there an algorithm 
   for any line-segment of finite length and any initial point 
   to localize  eventually the middle point of the line-segment 
 (see Problem~\ref{prob:original} for more precise)? 
  }
\label{Fprob}
\end{center}
\end{figure}
%%%%%%%%%%%%%%
 Unfortunately, 
   he is too impulsive a person to recall the previous position. 
% Actually she is putting a blindfold on her eyes. 
% Because she is an excellent balancer, 
%   she can move exactly the length which she wants. 
% However, it is too nervous work to recall the previous position. 
 Simply to say, 
  Walker does not have any memory about the previous motion. 
 Can he localize the midpoint of the rope?
 This is the question of this paper. 
%%%%%%%%%%%%%%%%%%%%%%%%%%%%%%%%%%
%\subsection{Formal description of the problem}\label{sec:formal}
 Precisely, 
  the question is formulated 
  as the existence of a transition map for a discrete time deterministic process, as follows. 
%%%%%%%%%%%%%%%%%%%%%%
\begin{problem}[Basic problem]\label{prob:original}
 A real $D$ ($1 < D < \infty$), and 
  a real $x$ such that $ -D \leq  x \leq D$ are given as an instance of the problem. 
 An {\em observation function} $\phi \colon \mathbb{R} \times [-D,D] \to \mathcal{O}$, 
 where $\mathcal{O}:= (\{R,L\} \times [0,D)) \cup \{\Myquit\}$, 
  is defined by 
\begin{eqnarray}
\phi(D,x) = 
\begin{cases}
 (R,D-x)& (\mbox{if $x > 0$})\\
 (L,D+x)& (\mbox{if $x < 0$})\\
 \Myquit & (\mbox{otherwise, i.e., $x = 0$}). 
\end{cases}
\label{def:phi}
\end{eqnarray}
%%%%%%%
 For convenience, 
   let $\phi(D,x)=(\SIDE(x),d_D(x))$ when $x \neq 0$. 
 A map $f\colon \mathcal{O} \to [-D,D]$ is a {\em transition map}\index{transition map}\footnote{
 Here, 
  $x \mapsto f(\phi(D,x))$ represents a transition of Walker on the interval $[-D,D]$. 
 } 
  if 
  $f(\phi(D,x))-x$ only depends on $d_D(x)$ and $\SIDE(x)$ 
    (but independent of $D$ or $x$), and  
  $f(\phi(D,x)) \in [-D,D]$ for any $x \in [-D,D]$. 
%%%%%%%%
 The goal is to design a transition map 
  $f \colon \mathcal{O} \to [-D,D]$ 
  for which an integer $n$ ($0 \leq n <\infty$) exists for any real $D$ ($1<D<\infty$) and $x_0 \in [-D,D]$ 
  such that $x_n=0$ 
   where $x_{i+1}=f(\phi(D,x_i))$ for $i=0,1,2,\ldots$. 
%%%%%
 More precisely, 
  let $\Psi \colon \mathbb{R} \times [-D,D] \to \mathbb{Z}_{\geq 0}$ be a {\em potential} function defined by 
\begin{eqnarray}
  \Psi(D,x) = \min\{ n \in  \mathbb{Z}_{\geq 0} \mid x_0=x,\ x_n=0,\ x_{i+1}=f(\phi(D,x_i))\}
  \label{def:Psi}
\end{eqnarray}
 for any instance $D$ ($1<D<\infty$) and $x \in [-D,D]$, 
  then $\Psi(D,x)$ needs to be bounded (may depend on $D$). 
\end{problem}
%%%%%%%%%%%%%%%%%%%%%%%%%%%%%%%%%%
 Here we make a brief remark. 
 Let $$M((\SIDE(x),d(x))) := f((\SIDE(x),d(x)))-x$$ for $(\SIDE(x),d(x)) \in \mathcal{O}$, 
 then $M((\SIDE(x),d(x)))$ represents the ``motion'' when Walker observes 
  the direction $\SIDE(x)$ and the distance $d(x)$. 
 Since only $\SIDE(x)$ and $d(x)$ are available to him,  
  $M((\SIDE(x),d(x)))$ (i.e., $f(\phi(D,x))-x$) must depend only on $\SIDE(x)$ and $d(x)$. 
%%%%%%%%%%%
 The potential function 
  $\Psi(D,x)$ represents the number of steps 
  to localize the midpoint by the algorithm given by the transition map~$f$.

%%%%%%%%%%%%%%%%%%%%%%%
\subsection{Our results}
 Contrary to its simple appearance of the problem, 
  it is far from trivial if the problem is solvable or not, and 
  the question will be left as unsettled in this paper. 
 This paper is concerned with 
   three variants of the problem minimally relaxing some constraints, and 
   shows the solvability of them by giving algorithms, respectively. 

%%%%%%%%%%%
 The first variant is a {\em convergence} version of the problem, 
   which relaxes the midpoint condition 
    such that Walker is required to reach {\em around} the midpoint, 
    instead of exactly at the midpoint (see Section~\ref{sec:converge} for detail). 
% Thus, the goal of the problem is to localize a point in the line-segment near the midpoint, 
%   instead of exactly localizing the midpoint.
 We give an algorithm, which uses a choice function. 

%%%%%%%%%%
 The second variant is about a prior knowledge on the length of the line segment. 
 The problem is (sometimes) solvable 
   if we assume some restriction on the length of a given line segment, 
   instead of arbitrary real (Section~\ref{sec:restriction}).
 For instance, 
   we prove that the problem is solvable 
   if the length of the line-segment is restricted to be an algebraic real (Section~\ref{sec:algebraic}). 
 Since the algebraic reals are at most countable many, 
  it is natural to ask if there is a solvable case 
%   we also show another solvable case, 
   where the length is restricted, but {\em uncountably} many. 
 We affirmatively answer the question with a simple example  (Section~\ref{sec:epsilon-omit}), 
  giving a technique different from Section~\ref{sec:algebraic} to solve it. 
 In fact, 
  the technique presented in Section~\ref{sec:epsilon-omit} is also applied, 
   with a carefully crafted modification, 
  to another interesting case: 
  the problem is solvable 
   if the length of the line-segment is assumed that 
   its fractional part is NOT in the {\em Cantor set} (Section~\ref{sec:nonCantor}). 
 We remark that the Lebesgue measure of the Cantor set is known to be zero, 
  meaning that  
  this example shows a bit more enhancement of solvable cases 
   from an uncountable set in Section~\ref{sec:epsilon-omit} to {\em almost everywhere} in Section~\ref{sec:nonCantor}. 

%%%%%%%%%%%
 The third variant allows Walker to have a small memory, 
   instead of no-memory (Section~\ref{sec:memory}). 
 As we stated, 
  if Walker uses some history of the previous position, then the problem is easily solved. 
 We show that {\em only a single-bit} of memory is sufficient to solve the problem.  
 The algorithm using some parity tricks is simpler, and could be more practical, 
   than the other cases which use choice functions. 

%%%%%%%%%%%%%
 Above arguments show that the problem is solvable with some reasonably minimal relaxations, 
    nevertheless we conjecture that the original problem is unsolvable. 
 Concerning the impossibility, 
   we also give an easy impossibility theorem, 
   where we assume that an algorithm is restricted to be mirror symmetric at the midpoint 
   (Section~\ref{sec:symmetric}).  
% In Concluding Remark ,  
 We also present some interesting versions unsettled (Section~\ref{sec:conclude}). 

\subsection{Background: computability of autonomous mobile robots}
%%%%%%%%%
 This paper is originally motivated by understanding the intractability of 
  {\em self-stabilization} of {\em autonomous oblivious mobile robots} with {\em limited visibility}. 
 Here, we briefly explain the background. 
%\paragraph{Autonomous mobile robots}
 Computability by autonomous mobile robots 
   is widely interested in distributed computing,  
   with many real applications 
    including wireless sensor networks, molecular robotics, and so on. 
 In the earliest work on the topic, 
  Suzuki and Yamashita~\cite{suz99} revealed that 
   no uniform algorithm can break the symmetry of two semi-synchronous anonymous robots on a 2-D plane. 
%  which is the first result on computability result of autonomous mobile robots in distributed computing. 
 In a recent work, 
  Yamauchi et al.~\cite{yama17} characterized the solvability of the plane formation problem, 
   that is the first characterization on the symmetry breaking of autonomous mobile robots in 3-D space. 
%%%%%%%%%%
% Motivated by real applications such as wireless sensor networks with mobile nodes, 
%  or by computability of a distributed system from the theoretical point of view. 

%%%%% 
 {\em Self-stabilization} is an important subject in distributed computing, 
   particularly from the view point of fault tolerance. 
 An algorithm is self-stabilizing 
   if it is proved eventually to lead a desired state (solution), from any initial state of the computing entity. 
 Mainly from the view point of self-stabilization, 
   robots in the literature have few memory, or often no memory (called {\em oblivious}), 
   and hence they are required to solve problems from ``geometric'' information observed. 
%%%%%
 Designing self-stabilizing distributed algorithms for mobile autonomous robots 
  has been intensively investigated %in distributed computing 
  on various problems such as 
  {\em pattern formation}~\cite{def02,fuj15,suz99,yam10,yama17,yama13}, 
  {\em gathering}~\cite{and99,and95,coh05,coh06,flo05}, 
  {\em self-deployment}~\cite{bar11,coh08,eft14,eft13,flo08,flo10,shibata16} including 
    {\em scattering} and 
    {\em coverage}. 

%%%%%%%%
 Mainly from the theoretical tractability, 
  robots are often assumed in the literature 
  to have (infinitely) large visibility such that they respectively observe the whole robots,  
 which  corresponds to a situation that sensor nodes are congested with in there sensing area 
  in the practical sense. 
%%%%%%%% 
 However, 
   real sensor nodes often do not have enough power of sensing, 
   and {\em limited visibility} is considered as a more practical model. % as well as a challenging target. 
%%%%%
 The limitation of the visibility causes many intractability, and 
  there are some developments \cite{and99,and95,def02,eft14,eft13,flo08,flo05,yama13} on the model, 
   but not so many. 
%    results are known 
  % not only in practice but also in theory, 
 In fact, the theoretical difficulty of the model is not well understood, yet.  
%  and it motives this work. 

%%%%%%%
 This paper is motivated 
  by the question why it is difficult to prove the intractability of 
  a self-stabilization of autonomous mobile robots with {\em limited visibility}. 
%%%%
 Problem~\ref{prob:original} may be an essence of the difficulty, 
  there is a single robot in a 1-D space and 
  the task is to localize the midpoint  
    where the visibility range of the robot is exactly half the length of the line segment. 
% Contrary to its simple appearance of the problem, 
%  it is far from trivial if the problem is solvable or not, and 
%  the question is unsettled yet. 

%%%%%%%%%%%%%%%%%%%%%%%%%%%
%\subsection{Related works}\label{sec:related-work}
\paragraph{Existing works about the autonomous mobile robots with {\em limited visibility}}
%%%%%%%%%%%%%
 Closely related problems, or a direct motivation of the paper, 
    are scattering or coverage over a line or a ring 
   by autonomous mobile robots with limited visibility~\cite{coh08,def02,eft14,eft13,flo08}. 
%%%%%%%%%%%
 Cohen and Peleg~\cite{coh08} were concerned with 
   spreading of autonomous mobile robots over a line (1-D space)
   where a robot observes the nearest neighboring robot in {\em each} of left and right side. 
 They presented a local algorithm leading to equidistant spreading on a line, and 
   showed convergence and convergence rate for fully synchronous (FSYNC) and semi-synchronous (SSYNC) models. 
 They also gave an algorithm to solve exactly 
    when each robot has enough size of memory, that is linear to the number of robots. 
%%%%%%%%%%%
 Eftekhari et al.~\cite{eft13}  
   studied the coverage of a line segment  %, and  
    by autonomous mobile robots placing grid points with minimum visibility to solve the problem. 
 They gave two local distributed algorithms, 
   one is for oblivious robots and it terminates in time quadratic to the number of robots, while 
   the other is for robots with a constant memory and it terminates in linear time. 
%%%%%%%%%%%
 Eftekhari et al.~\cite{eft14} 
    showed the impossibility of the coverage of a line segment (whose length is integral) by robots with limited visibility in SSYNC model 
    when robots do not share left-right direction. 
 Whereas, they showed that it is solvable even in asynchronous (ASYNC) model 
  if robots 
     shares left-right direction, 
     have a visibility range strictly greater than mobility range, and 
     know the size of visibility range. 

%%%%%%%%%%%
 Flocchini et al.~\cite{flo08} 
    were concerned with equidistant covering of a {\em circle} by oblivious robots with limited visibility. 
 They showed the impossibility of exact solution if they do not share a common orientation of the ring. 
 They also showed the possibility 
    by oblivious asynchronous robots with almost minimum visibility 
    when robots share a common orientation. 
%%%%%%%%%%%
 Defago and Konagaya~\cite{def02} 
   were concerned with circle formation {\em in 2-D space} 
    by oblivious robots with limited visibility, 
    where robots do not know the size of their visibility range. 
 In the paper, 
   they also dealt with equidistant covering of a circle, and 
   gave an algorithm for convergence. 

%%%%%%%%%%%%%%%%%%%%%%%
\subsection{Organization}
 This paper is organized as follows. 
 Sections~\ref{sec:converge}--~\ref{sec:memory} 
  respectively show the solvability of three variants of Problem~\ref{prob:original}, 
   minimally relaxing some constraints. 
%%%%%%%
 Section~\ref{sec:converge} is concerned with a {\em convergence} version of the problem. 
%   which relaxes the visibility condition 
%    such that we observe the both ends {\em around} the midpoint, 
%    instead of exactly at the midpoint. 
 Section~\ref{sec:restriction} 
   assumes {\em a prior knowledge} on the length of the line segment, 
   instead of arbitrary real.
 Section~\ref{sec:memory} 
   shows that the problem is solvable if we have a {\em single-bit} of memory.  
%%%%%%%%%%%%%
 Concerning the impossibility, 
  Section~\ref{sec:symmetric} gives an easy impossibility theorem, 
   where we assume that an algorithm is restricted to be mirror symmetric at the midpoint. 
%%%%%%%%%%%%%
 Section~\ref{sec:conclude} concludes this paper, 
  with some open problems. 

%%%%%%%%%%%%%%%%%%%%%%%%%%%%%%%%%%%%%%%%%%%%%%%%%%%%%%%
\section{Relaxation 1: Convergence }\label{sec:converge}
%%%%%%%%%%%
 This section shows that the {\em convergence} version of Problem~\ref{prob:original} is solvable. 
 Precisely, we are concerned with the following problem  
%%%%%%%%%%%%%%%%%%%%%%
\begin{problem}[Convergence]\label{prob:converge}
 A real $D$ ($1 < D < \infty$), \underline{a real $\epsilon$ ($0 < \epsilon \leq D$)}
  and a real $x \in [-D,D]$ are given as an instance of the problem. 
 The observation function (of Problem~\ref{prob:converge}) $\phi \colon \mathbb{R} \times \mathbb{R} \times [-D,D] \to \mathcal{O}$ 
  is given by 
\begin{eqnarray}
\phi(D,\epsilon,x) = 
\begin{cases}
 (R,D-x)& (\mbox{if $x > \epsilon$})\\
 (L,D+x)& (\mbox{if $x < -\epsilon$})\\
 \Myquit & (\mbox{otherwise, i.e., $ -\epsilon \leq x \leq \epsilon$}). 
\end{cases}
\end{eqnarray}
%%%%%%%%
 A map $f$ is a transition map if $f(\phi(D,x))-x$ depends only on $d_D(x)$ and $\SIDE(x)$
 (but independent of $D$, $x$, \underline{or $\epsilon$}), and 
  $f(\phi(D,x)) \in [-D,D]$ for any $x \in [-D,D]$. 
 The goal of the problem is to design a transition map 
  $f \colon \mathcal{O} \to [-D,D]$ 
  for which an integer $n$ ($0 \leq n <\infty$) exists 
  for any reals $D$ ($1<D<\infty$), $\epsilon$ ($0 < \epsilon \leq D$) and $x_0 \in [-D,D]$ 
  such that $-\epsilon \leq x_n \leq \epsilon$ 
   where $x_{i+1}=f(\phi(D,\epsilon,x_i))$ for $i=0,1,2,\ldots$. 
\end{problem}
%%%%%%%%%%%%%%%%%%%%%%%%%%%%%%%%%%%%%%%%%%%%%%
 The condition that $f(\phi(D,x))-x$ is independent of $\epsilon$ 
   corresponds to the situation that $\epsilon$ is not available to Walker. 
 Thus, 
   an algorithm is required two conflicting functions: 
 The step-lengths are (preferably) decreasing, otherwise Walker misses the small interval $[-\epsilon,\epsilon]$. 
 On the other hand, 
    the total length of the moves should diverge as increasing the number of steps, 
  otherwise Walker stops before reaching at the midpoint 
    when $D$ is larger than the upper bound of the total length of the moves. %, like ``Achilles and the Tortoise.'' 

%%%
 The rest of this section shows the following theorem. 
%%%%%%%%%%%%%%%%%%%%%%%%%%%%%%%%%%%%%%%%%%
\begin{theorem}\label{thm:converge}
 Problem~\ref{prob:converge} is solvable. 
\end{theorem}
%%%%%%%%%%%%%%%%%%%%%%%%%
\subsection{Preliminary}
%%%%%%%%%%%%
 As a preliminary step of the proof of Theorem~\ref{thm:converge}, 
  we briefly remark the following three propositions on the {\em reciprocal of primes}.
  %which are versions of well-known fundamental facts. 
%%%%%%
 Let $\mathbb{P}$ denote the whole set of prime numbers, and 
 let $\pi_i \in \mathbb{P}$ $(i=1,2,3,\ldots)$ denote the $i$-th smallest prime number, 
 i.e., $\pi_1=2$, $\pi_2=3$, $\pi_3=5$, $\pi_4=7$, \ldots. 
  
%%%%%%%%%%%%%%%%%%%%%%%%%
 For convenience of the later argument, let 
\begin{equation}
  \Delta_k = 
   \left\{\frac{\pm n}{\prod_{i=1}^{k}{\pi_i}}\ \middle|\ 
    n \in \mathbb{Z}_{>0} \mbox{ such that } 
%    n \in \mathbb{Z}\backslash \{0\} \mbox{ such that } 
%    \forall k \in \{1, 2, \ldots , i\},\ 
%    \frac{n}{\pi_k} \notin \mathbb{Z}
 \pi_i \nmid n\ (\forall i \in \{1,2,\ldots,k\})
\right\}
\end{equation}
 where $a \nmid b$ denotes for $a,b \in \mathbb{Z}_{>0}$ that $a$ is not a divisor of $b$, i.e., $b/a \not\in \mathbb{Z}$. 
  for each $k \in \mathbb{Z}$.
 We also define $\Delta_0 = \{0\}$, for convenience.
% Notice that $\Delta_i \cap \Delta_j = \emptyset$ when $i \ne j$.
%%%%%%%%%%%
 We make three remarks on~$\Delta_k$. 
%%%%%%%%%%%%%%%%%%
\begin{proposition}\label{prop0}
$\Delta_i \cap \Delta_j = \emptyset$ when $i \ne j$.
\qed\end{proposition}
%%%%%%%%%%%%%%%%%%%
%%%%%%%%%%%%%%%%%%
\begin{proposition}\label{coro2}
 If $\delta \in \Delta_k$ $(k \in \mathbb{Z}_{\geq 0})$ and $m \in \mathbb{Z}$ 
  then $\delta+m \in \Delta_{k}$. 
\qed\end{proposition}
%%%%%%%%%%%%%%%%%%%
%%%%%%%%%%%%%%%%%%
\begin{proposition}\label{coro1}
 If $\delta \in \Delta_k$ $(k \in \mathbb{Z}_{\geq 0})$ 
 then $\delta + \frac{1}{\pi_{k+1}} \in \Delta_{k+1}$. 
\qed\end{proposition}
%%%%%%%%%%%%%%%%%%%

 We also remark the following fact, 
   which is easily derived from the classical fact due to due to Euler~\cite{Euler1737}, 
  that the sum of the reciprocals of all prime numbers diverges. 
%%%%%%%%%%%%%%%%%%%%%%%%%
\begin{proposition}\label{prop1}
 $\sum_{i=j}^{\infty}{\frac{1}{\pi_i}} = \infty$ for any finite $j \in \mathbb{Z}_{> 0}$.
\qed\end{proposition}
%%%%%%%%%%%%
%\begin{proof}
% It is known, due to Euler~\cite{Euler1737}, 
%  that %the sum of the reciprocals of all prime numbers diverges, i.e., 
%  $\sum_{i=1}^{\infty}{\frac{1}{\pi_i}} = \infty$ (cf.~\cite{sosu}).
% Since the finite sum $\sum_{i=1}^{j-1}{\frac{1}{\pi_i}}$ is upper bounded for any finite $j$, 
%  we obtain the claim. 
%\end{proof}
%%%%%%%%%%%%%%%%%%%%%%%%%%

%%%%%%%%%%%%%%%
\subsection{Proof of Theorem~\ref{thm:converge}}
 Now, we prove Theorem~\ref{thm:converge}. 
\begin{proof}[Proof of Theorem~\ref{thm:converge}]
 The proof is constructive. 
 We define a transition map $f\colon\mathcal{O} \to [-D,D]$ to solve Problem~\ref{prob:converge} by 
\begin{eqnarray*}
 f((L,d))& =& 
\begin{cases}
 x+\dfrac{1}{\pi_{k+1}} 
 & (\mbox{if $d \in \Delta_k$  for some $k \in \mathbb{Z}_{\geq0}$}\footnotemark[2]), \\[2ex]
% & (\text{if $d \in \Delta_k$  for some $k \in \mathbb{Z}_{\geq0}$\footnote{
% Notice that $k$ is uniquely determined by Proposition~\ref{prop0} if it exists.
%}}), \\[2ex]
 x-d 
 \quad (=-D) 
 & (\mbox{otherwise, i.e., $ d \not\in \Delta_k$  for any $k \in \mathbb{Z}_{\geq0}$}), 
\end{cases} \\
 f((R,d))& =& x-1,   \\
f(\Myquit)&=&x
\end{eqnarray*}
\footnotetext[2]{
 Notice that $k$ is uniquely determined by Proposition~\ref{prop0}, if it exists.
 }in each case of $\phi(D,x) = (L,d)$, $(R,d)$ or $\Myquit$
  for any $x \in [-D,D]$ (see also Algorithm~\ref{alg:converge} and Figure~\ref{Falg1R}). 
 It is not difficult to observe that 
   $f$ is a transition map (recall Problem~\ref{prob:converge}). 
%%%%%%%%%%%%%%%%%%
 Then, we show for any $x_0 \in [-D,D]$ that 
  there is a finite $n \in \mathbb{Z}_{\geq 0}$ such that $-\epsilon \leq x_n \leq \epsilon$ 
  where $x_t = f(\phi(D,x_{t-1}))$ for $t=1,2,\ldots$. 
 For convenience, 
  let $(\SIDE(t),d(t))=\phi(D,x_t)$\footnote{
  In the rest of the paper, 
   we use $\SIDE(t)$ and $d(t)$ 
    respectively as abbreviations of $\SIDE(x_t)$ and $d(x_t)$, for readability. 
  }. 

%%%%%%%%%%%%%%%%%%
 Firstly, we observe that if $\SIDE(t)=R$, then 
   there exists $t'$ ($t'>t$) such that $\SIDE(t')=L$, or $\phi(D,x_{t'})=\Myquit$, 
   since the sum of $-1$'s diverges (to $-\infty$). 
 Secondly, we observe that if $\SIDE(t)=L$ and $d(t) \not\in \Delta_k$ for any $k=0,1,\ldots$, 
  then $\SIDE(t+1)=L$ and $d(t+1)=0 \in \Delta_0$. 
 Thus, without loss of generality, 
  we may assume that $\SIDE(0)=L$ and $d(0) \in \Delta_k$ for some $k=0,1,\ldots$, 
  where notice that $k$ is uniquely determined by Proposition~\ref{prop1}. 

 Next, 
   we claim that 
   if $t,t'$ ($t<t'$) satisfies $\SIDE(t)=\SIDE(t')=L$ and $d(t) \in \Delta_k$, 
   then $d(t') \in \Delta_{k'}$ and $k' > k$, 
   i.e., the index $k$ of $\Delta_k$ is monotone increasing for time $t$. 
 If $\SIDE(t+1)=L$, then it is not difficult to see that $d(t+1) =d(t)+1/\pi_{k+1}$ by the definition of $f$. 
 Proposition~\ref{coro1} (and Proposition~\ref{prop0}) implies that $d(t+1) \in \Delta_{k+1}$. 
 If $\SIDE(t+1)=R$, then $x_{t+2}=x_{t+1}-1=x_t+1/\pi_{k+1}-1$ by the definition of $f$. 
 Since $1/\pi_{k+1} \leq 1$, $x_{t+2} < x_t$ holds, and hence $\SIDE(t+2)=L$. 
 Furthermore, notice that $d(t+2) = d(t)+1/\pi_{k+1}-1$, 
  thus Proposition~\ref{coro1} and~\ref{coro2} (and Proposition~\ref{prop0}) imply that $d(t+2) \in \Delta_{k+1}$. 
 By the arguments we obtain the claim. 
 
 Now, it is not difficult to observe that 
  Walker eventually gets in $\Myquit$. 
 Since 
  the index $k$ of $\Delta_k$ is non-decreasing for time $t$, 
  suppose $\SIDE(t)=L$ and $d(t) \in \Delta_k$ where $k$ satisfies $1/\pi_k < 2\epsilon$. 
 Then, there is $n$ $(n > t)$ such that $x_n \geq -\epsilon$ 
  since $\sum_{j=k}^{\infty} 1/\pi_j$ diverges by Proposition~\ref{prop1}. 
 Suppose for convenience that $\SIDE(n-1)=L$, i.e., $x_{n-1}<-\epsilon$. 
% By the assumption of the case that $1/\pi_k < 2\epsilon$, 
 Then, 
  $x_n \leq x_{n-1}+1/\pi_k < x_{n-1} + 2\epsilon < \epsilon$.  
%  where the last inequality follows that $x_{n-1}<-\epsilon$. 
 This implies $-\epsilon \leq  x_n \leq \epsilon$. 
\end{proof}

%%%%%%%%%Algo.1right%%%%%%%%%%%%
\begin{figure}[tbp]
\begin{center}
\includegraphics[width=60mm]{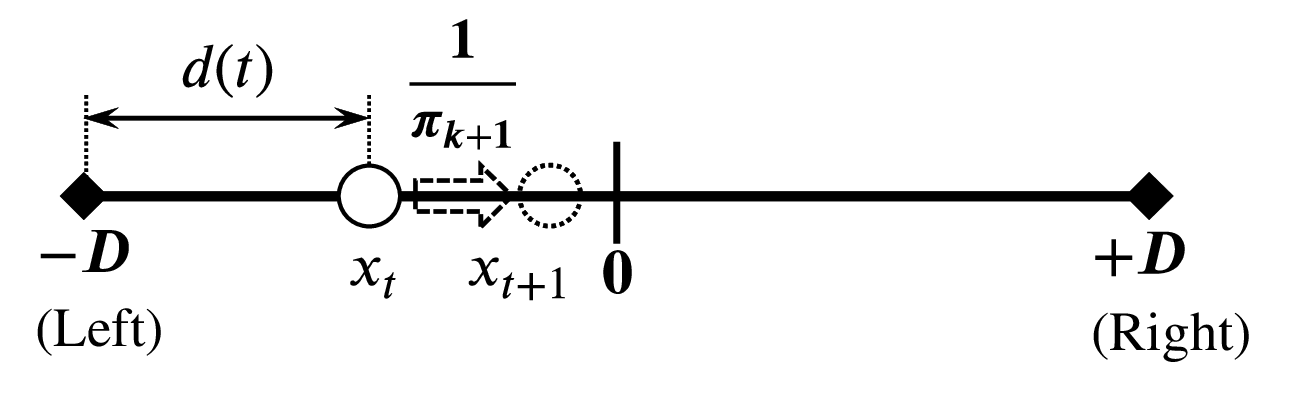}
\hspace{3em}
\includegraphics[width=60mm]{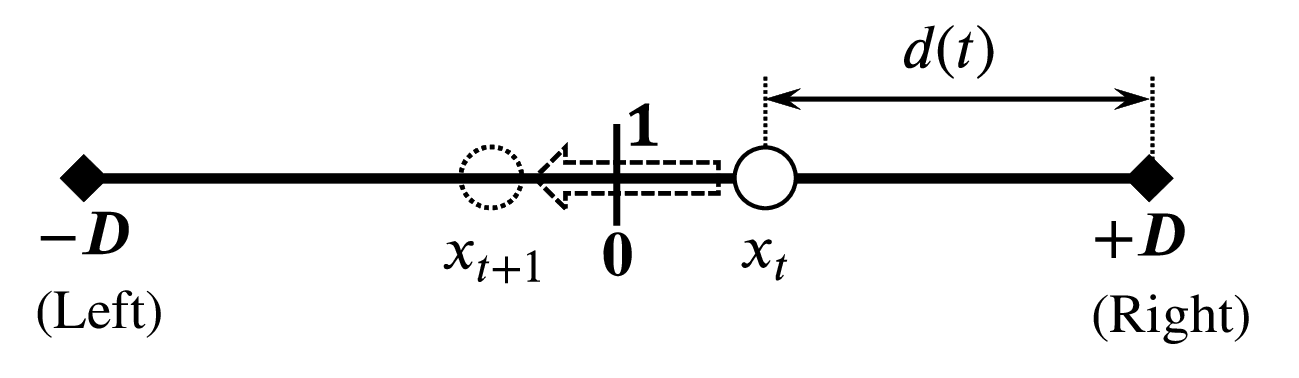}
\caption{
 For Algorithm~\ref{alg:converge}:  
 left figure shows the case of $\SIDE(t)=L$ and 
 right figure shows the case of $\SIDE(t)=R$. 
}
\label{Falg1R}
\end{center}
\end{figure}

%%%%%%%%%%%%%%%%%%%%%%%
\begin{algorithm}
\caption{(for convergence)}
\label{alg:converge}
\algsetup{indent=1.5em}
\begin{algorithmic}[1]
\LOOP
\STATE {observe $(\SIDE,d)$ or $\Myquit$}
	\IF {$\SIDE=L$}
		\IF {$d \in \Delta_k$}
		  \STATE move right by $\frac{1}{\pi_{k+1}}$
		\ELSE
	  	  \STATE move to the left-end 
		\ENDIF
  	\ELSIF {$\SIDE=R$}
		\STATE move left by $1$
	\ELSE 
	  \STATE (i.e., $\Myquit$ is observed) stay there
	\ENDIF
\ENDLOOP
\end{algorithmic}
\end{algorithm}

%%%%%%%%%%%%%%%%%%%%%%%%%%%%%%%%%%%%%%%%%%%%%%%%%%%%%%%%%%%
\section{Relaxation 2: Restricted $D$}\label{sec:restriction}
%%%%%%%%%%%%%
 Problem~\ref{prob:original} is solved if we have some prior knowledge on $D$. 
 This section presents some nontrivial and interesting examples. 
 Section~\ref{sec:algebraic} proves that 
   the problem is solvable if $D$ is restricted to be an algebraic real. 
 While the algebraic reals are at most countable many, 
  Section~\ref{sec:epsilon-omit} and~\ref{sec:nonCantor} respectively show  
  other interesting cases where $D$ is assumed to be in some {\em uncountable} sets. 
%%%%%%%%%%%%%%%%%%%%%%%%%%%%%%%%%%%%%%%%%%%%%%%%%%%%%%%%%%%
\subsection{Algebraic real $D$}\label{sec:algebraic}
%%%%%%%%%%%%%
 A real $r \in \mathbb{R}$ is {\em algebraic} 
  if $f(r)=0$ holds for a polynomial $f$ of rational coefficients. 
 Let $\mathbb{A}$ denote the whole set of algebraic reals. 
 We are concerned with the following problem. 
%%%%%%%%%%%%%%%%%%%%%%
\begin{problem}[Algebraic real $D$]\label{prob:algebraic}
 As given 
  the observation function $\phi \colon \mathbb{R} \times [-D,D] \to \mathcal{O}$ defined by \eqref{def:phi}, 
%%%%%%%%
 the goal of the problem is to design a transition map 
  $f \colon \mathcal{O} \to [-D,D]$ 
  for which the potential function $\Psi(D,x)$, defined by \eqref{def:Psi}, is bounded 
  for any \underline{algebraic real} $D$ ($1 < D < \infty$) and any real $x \in [-D,D]$. 
\end{problem}
%%%%%%%%%%%%%%%%%%%%%%%%%%%%%%%%%%%%%%%%%%
\begin{theorem}\label{thm:rational}
 Problem~\ref{prob:algebraic} is solvable. 
\end{theorem}

 For convenience, we define 
\begin{equation}
 \Sigma_k = \left\{a-\frac{k}{{\rm e}} \ \middle|\ a \in \mathbb{A} \right\}
\end{equation}
  for any $k \in \mathbb{Z}_{\geq 0}$, where ${\rm e}$ denotes Napier's constant. 
 Using the well-known fact that ${\rm e}$ is transcendental, 
  we claim the following lemma as a preliminary step of the proof of Theorem~\ref{thm:rational}. 
%%%%%%%%%%%%%%%%
\begin{lemma}\label{algebraic}
 $\Sigma_i \cap \Sigma_j = \emptyset$ for $i \ne j$.
\end{lemma}
%%%%%%%%%%%%
\begin{proof}
 Assume for a contradiction that $r \in \Sigma_i \cap \Sigma_j$ exists. 
 Since $r \in \Sigma_i$, 
  there is an algebraic real $a_1\in \mathbb{A}$ such that $r = a_1 - \frac{i}{\rm e}$. 
 Similarly,  
  there is an algebraic real $a_2\in \mathbb{A}$ such that $r = a_2 - \frac{j}{\rm e}$ 
  since $r \in \Sigma_j$. 
 Thus, we obtain that $a_1-a_2 - \frac{i-j}{\rm e} = 0$, 
  which implies ${\rm e}=\frac{a_1-a_2}{i-j}$ where we specially remark that $i-j \neq 0$. 
% Since $\mathbb{A}$ is a field, 
 Notice that $\frac{a_1-a_2}{i-j} \in \mathbb{A}$, and  
 it contradicts to that ${\rm e} \not\in \mathbb{A}$. 
\end{proof}

%%%%%%%%%%%%%%%%%%%%%%%%%%%%%%%%%%%%%
%\subsection{Proof of Theorem~\ref{thm:rational}}
%%%%%%%%%%%%%%%%%%%%%%%%%%%%%%%%%%%
 Now, we prove Theorem~\ref{thm:rational}.
%%%%%%%%%%%%%%%%
\begin{proof}[Proof of Theorem~\ref{thm:rational}]
 The proof is constructive. 
 We define a transition map $f\colon\mathcal{O} \to [-D,D]$ to solve Problem~\ref{prob:algebraic} by 
\begin{eqnarray*}
f((L,d))&=&\begin{cases}
  x+\tfrac{1}{{\rm e}} 
    & \mbox{if $d = \tfrac{k}{{\rm e}}$ for some $k = 0,1,2,\ldots$} \\
  x-d \quad (=-D) 
    & \mbox{otherwise, i.e., $d \neq \frac{k}{{\rm e}}$ for any $k \in \mathbb{Z}_{\geq 0}$}\\
 \end{cases} \\
f((R,d))&=&\begin{cases}
  x- \frac{1}{2}\left(\frac{k}{{\rm e}} -d \right)
    & \mbox{if $d + \frac{k}{{\rm e}} \in \mathbb{A}$ for some $k =1,2,\ldots$, and 
    $0 < \frac{1}{2}\left(\frac{k}{{\rm e}} -d \right) < d$} \\[2ex]
  x-d 
    & \mbox{otherwise} 
 \end{cases}  \\
f(\Myquit) &=& x
\end{eqnarray*}
   in each case of $\phi(D,x) = (L,d)$, $(R,d)$ or $\Myquit$
  for any $x \in [-D,D]$ (see also Algorithm~\ref{alg:rational}). 
%  where $e$ is Napier's number.
 It is not difficult to observe that 
   $f$ is a transition map (recall Problem~\ref{prob:original}). 
 For convenience, 
  let $(\SIDE(t),d(t))=\phi(D,x_t)$. 

%%%%%%%%%%%%%%%%%%
 First, we show for any $x_0 \in [-D,0)$ that 
  a finite $n \in \mathbb{Z}_{> 0}$ exists such that $x_n=0$ 
  where $x_t = f(\phi(D,x_{t-1}))$ for $t=1,2,\ldots$. 
 If $d(0) \neq k/{\rm e}$ for any $k \in \mathbb{Z}_{\geq 0}$, 
  then $x_1 = -D$, meaning that $d(1)=0$, thus it is reduced to the case $d(0) = k/{\rm e}$ for $k=0$. 
% We also remark that 
%   $\SIDE(t)=L$ and $d(t)=k/{\rm e}$ imply that $x_t = -D+k/{\rm e}$. 
 Suppose that $\SIDE(0)=L$ and $d(0) = k/{\rm e}$ for some $k$. 
% , meaning that $x_t=-D+ k/{\rm e}$. 
 Then, %$x_{t+1} = -D + k/{\rm e} + 1/{\rm e} = -D+(k+1)/{\rm e}$, 
  it is not difficult to observe that $d(t+1) = (k+1)/{\rm e}$. 
 This inductively implies that 
   we have a finite $\tau = \min\{t' \in \mathbb{Z}_{>0} \mid \SIDE(t') = R\}$
    since $\lim_{i \to \infty}i/{\rm e}=\infty$. 
 Notice that $x_{\tau} = -D+k'/{\rm e}$ 
  for some ${k'} \in \mathbb{Z}_{>0}$. 
 Since the hypothesis that $D \in \mathbb{A}$, 
  $d(\tau) = D-x_{\tau} = 2D-k'/{\rm e} \in \Sigma_{k'}$ holds, 
  in other words $d(\tau)+k'/{\rm e} = 2D \in \mathbb{A}$ holds 
  where we specially remark that ${k'}$ is uniquely determined by Lemma~\ref{algebraic}. 
%%%%
% Furthermore, 
%   $d(\tau)+k'/{\rm e} = (D-x_{\tau})+k'/{\rm e} = 2D \in \mathbb{A}$ 
%   by the hypothesis $D \in \mathbb{A}$.
 Now, it is not difficult from the definition of $f$ to observe that 
  $x_{\tau+1} = x_{\tau}-\frac{k'/{\rm e}-d(\tau)}{2}=0$ holds. 
 Here we also remark that 
  $\frac{k'/{\rm e}-d(\tau)}{2} < d(\tau)$ holds at that time 
  since 
   $\frac{k'/{\rm e}-d(\tau)}{2} = x_{\tau} < 1/{\rm e}$ holds 
  while $d(\tau) = D - x_{\tau} > 1 - 1/{\rm e}$ holds. 
 We obtain the claim in this case (Fig~\ref{Falg2}).  
 
 Next, we are concerned with the case that $x_0 \in (0,D]$, and 
  show that there is $t \in \mathbb{Z}_{>0}$ such that $x_t \leq 0$, 
  then it is reduced to the case that $x_0 \in [-D,0)$, or the trivial case $x_0=0$. 
 Notice that 
  if $d(s)+k/{\rm e} \not \in \mathbb{A}$ for any $k \in \{1,2,\ldots\}$ then $d(s+1)=2d(s)$, 
  which implies that if the case occurs at most finite times, 
  we eventually obtain the desired case that $x_t<0$. 
%%%%%%%
% In fact, we claim that the case of $d(s)+k/{\rm e} \in \mathbb{A}$ occurs at most once before $x_t<0$. 
% Without loss of generality, we may assume that 
 Suppose $d(0)+k/{\rm e} \in \mathbb{A}$, 
  then we claim that $d(s)+i/{\rm e} \not\in \mathbb{A}$ 
  for any $s \in \{t \in \mathbb{Z}_{>0} \mid \forall t'\leq t,\ x_{t'} > 0\}$ 
  and for any $i \in \mathbb{Z}_{>0}$. 
 By the definition of $f$, 
  if $0 < \frac{k/{\rm e}-d(0)}{2} < d(0)$ then 
%\begin{eqnarray*}
%  d(1) 
%   &=& D-x_1 \\
%   &=& D - \left(x_0 - \frac{k/e-d(0)}{2} \right) \\
%   &=& d(0) + \frac{k/e-d(0)}{2} \\
%   &=& \frac{d(0)+k/e}{2} 
%\end{eqnarray*} 
$
 d(1) 
= D-x_1
= D - \left(x_0 - \frac{k/{\rm e}-d(0)}{2} \right) 
= d(0) + \frac{k/{\rm e}-d(0)}{2} 
= \frac{d(0)+k/{\rm e}}{2} 
$
 and hence $d(1) \in \mathbb{A}$ by the hypothesis of the case. 
 This implies that $d(1) + i/{\rm e} \not\in \mathbb{A}$ for any $i = 1,2,3,\ldots$. 
 Clearly, $d(2) = 2d(1) \in \mathbb{A}$, and recursively we obtain the claim. 
\end{proof}
%
%
%
%%%%%%%%%Algo.2%%%%%%%%%%%%
\begin{figure}[tp]
\begin{center}
\includegraphics[width=70mm]{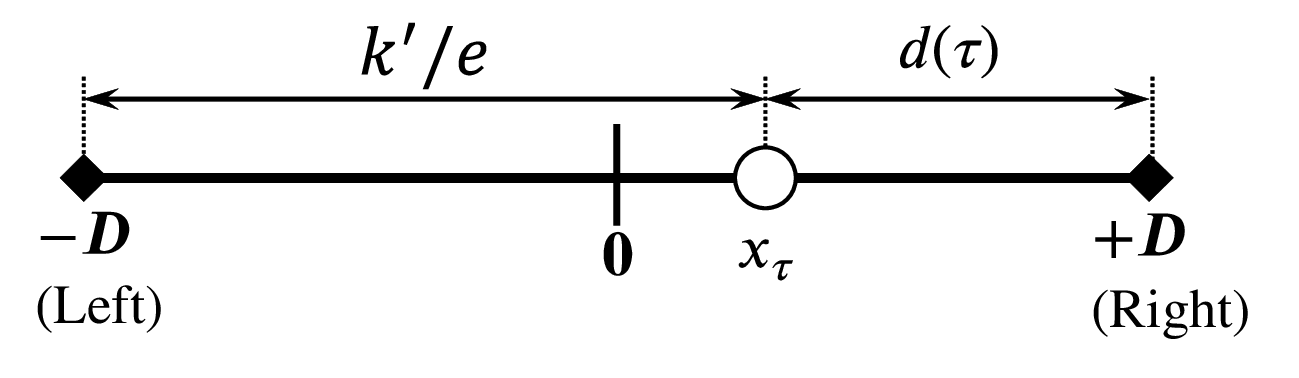}
\caption{
At time $\tau$ in the proof of Theorem~\ref{thm:rational}. 
}
\label{Falg2}
\end{center}
\end{figure}
%
%
%
%%%%%%%%%%%%%%%
\begin{algorithm}
\caption{(for algebraic $D$)}
\label{alg:rational}
\algsetup{indent=1.5em}
\begin{algorithmic}[1]
\LOOP
\STATE observe $(\SIDE,d)$ or $\Myquit$
	\IF {$\SIDE=L$}
		\IF {$d = k/{\rm e}$ for some $k=0,1,2,\ldots$}
		  \STATE move right by $1/{\rm e}$
		\ELSE
   	  	  \STATE move to the left-end 
		\ENDIF
	\ELSIF {$\SIDE=R$}
		\IF {[$d+k/{\rm e} \in \mathbb{A}$ for some $k=1,2,\ldots$] and [$\frac{k/{\rm e}-d}{2} \leq d$]}
	  	  	  \STATE move left by $\frac{k/{\rm e}-d}{2}$
		\ELSE
		  \STATE move left by $d$ 
		\ENDIF
	\ELSE
	  \STATE (i.e., $\Myquit$ is observed) stay there
	\ENDIF
\ENDLOOP
\end{algorithmic}
\end{algorithm}

%%%%%%%%%%%%%%%%%%%%%%%%%%%
\subsection{$D$ chosen from an uncountable set}\label{sec:epsilon-omit}
%\section{$D$ of $\aleph_1$}\label{sec:epsilon-omit}
%%%%%%%%%%%%%%%%%%%%%%%%%%%%
 In Section~\ref{sec:algebraic}, 
   we have shown that the localization problem is solved if $D$ is algebraic. 
 Since the algebraic reals are at most countable many, 
  it is natural to ask 
   if there is a solvable case where $D$ is chosen from an uncountable set. 
 Section~\ref{sec:epsilon-omit} affirmatively answers the question. 

%%%%%%%%%%%%%%%%%
%\subsubsection{Problem}\label{sec:epsilon}	
 For an arbitrary small $\epsilon$ ($0 < \epsilon < 1/2$), let 
\begin{eqnarray}
 \overline{\mathbf{S}}_{\epsilon} = \{x \in \mathbb{R} \mid x - \lfloor x \rfloor \in [\epsilon,1-\epsilon] \}. 
\end{eqnarray} 
 Then, we are concerned with the following problem. 
\begin{problem}[$D$ in $\overline{\mathbf{S}}_{\epsilon}$]\label{prob:epsilon-omit}
 As given 
  the observation function $\phi \colon \mathbb{R} \times [-D,D] \to \mathcal{O}$ defined by \eqref{def:phi}, 
%%%%%%%%
 the goal of the problem is to design a transition map 
  $f \colon \mathcal{O} \to [-D,D]$ 
  for which the potential function $\Psi(D,x)$, defined by \eqref{def:Psi}, is bounded 
  for any $D$ ($1 < D < \infty$) such that \underline{$D$ in $\overline{\mathbf{S}}_{\epsilon}$} and any real $x \in [-D,D]$. 
\end{problem}
%%%%%%%%%%%%%%%%%%%%%%%%%%%%%%%%%%%%%%%%%%
\begin{theorem}\label{thm:epsilon-omit}
 Problem~\ref{prob:epsilon-omit} is solvable. 
\end{theorem}

 Notice that the cardinality of $\overline{\mathbf{S}}_{\epsilon}$ is equal to that of $\mathbb{R}$.
 Furthermore, the Lebesgue measure of $[\epsilon,1-\epsilon]$ is equal to $1-2\epsilon$, and hence 
 $\frac{\overline{\mathbf{S}}_{\epsilon} \cap [0,1]}{[0,1]} = 1-2\epsilon$.

%%%%%%%%%%%%%%%%%%%%%%%%%%%%%%%%%%%%%%%%%%
\subsubsection{Idea for an algorithm}
%%%%%%%%%%%%%%%%%%
%%%%%%%%%Algo.2%%%%%%%%%%%%
\begin{figure}[tp]
\begin{center}
\includegraphics[width=140mm]{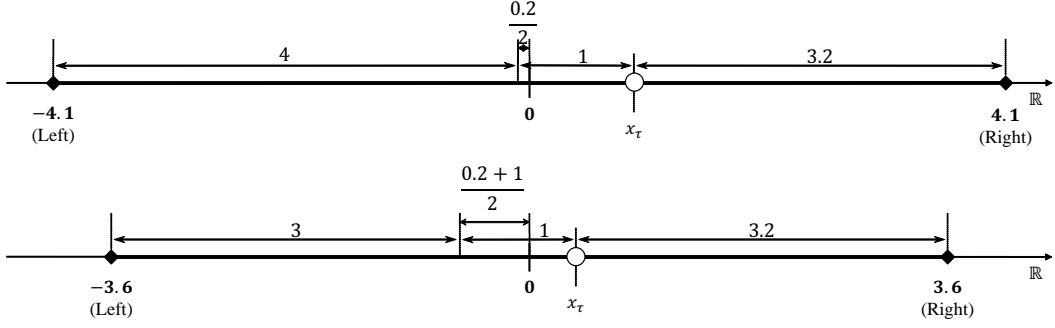} 
\caption{
 An example of the basic idea to solve Problem~\ref{prob:epsilon-omit}. 
 Suppose that 
  Walker steps one by one from the left-end, and 
  he observes $\SIDE(\tau) = R$ for the first time at $x_{\tau}$.
 Notice that the number of steps $\tau = \lfloor D \rfloor + 1$ from the left-end to $x_{\tau}$, and 
  hence $d(\tau) = 2D - (\lfloor D \rfloor + 1)$ holds. 
 Suppose $d(\tau)=3.2$, 
  then there are two possibilities that $D=4.1$ or $3.6$; 
 If $D - \lfloor D \rfloor < 1/2$ 
  then $\lfloor d(\tau) \rfloor = \lfloor D \rfloor - 1 = 3$ 
     by Proposition~\ref{obs:epsilon-omit1}, and 
  Lemmas~\ref{lem:epsilon-omit-midpoint} implies that  
  $ D = \lfloor D \rfloor + \tfrac{d - \lfloor d \rfloor}{2} = 4+ \tfrac{3.2 -3}{2} = 4.1$ holds. 
 If $D - \lfloor D \rfloor \geq 1/2$ 
  then $\lfloor d(\tau) \rfloor = \lfloor D \rfloor = 3$ 
     by Proposition~\ref{obs:epsilon-omit1}, and 
  Lemmas~\ref{lem:epsilon-omit-midpoint} implies that  
  $ D = \lfloor D \rfloor + \tfrac{d - \lfloor d \rfloor}{2} +\frac{1}{2} = 3+ \tfrac{3.2 -3}{2} + 0.5 = 3.6$ holds. 
}
\label{fig:3.1.1a}
\end{center}
\end{figure}
%%%%%%%%%%%%%%%%
 Here we briefly explain the basic idea to solve Problem~\ref{prob:epsilon-omit} 
   (see also Figures~\ref{fig:3.1.1a} and~\ref{fig:3.1.1b}). 
 As an easy case, suppose that Walker starts from the left-end.  
 In our algorithm, 
  Walker iteratively moves to right with length one in a step unless he observes the right-end. 
 When Walker observes the right-end for the first time, 
  Walker is in the distance $\lfloor D \rfloor + 1$ from the left-end. 
 At that time, the distance $d$ from the right-end, which Walker observes, is equal to $2D-(\lfloor D \rfloor + 1)$. 
 Here we observe the following easy but important proposition. 
\begin{proposition}\label{obs:epsilon-omit1}
% Suppose $D - \lfloor D \rfloor \neq 0$ holds. Then, 
 For any $D \in \mathbb{R}$, 
\begin{eqnarray*}
 \lfloor 2D - (\lfloor D \rfloor + 1) \rfloor =
\begin{cases}
 \lfloor D \rfloor -1 &  (\mbox{if $D - \lfloor D \rfloor < \frac{1}{2}$}) \\
 \lfloor D \rfloor    &  (\mbox{otherwise})
\end{cases}
\end{eqnarray*}
holds.
% If $D - \lfloor D \rfloor < 1/2$, then $ \lceil 2D - (\lfloor D \rfloor + 1) \rceil = \lceil D \rceil+1$ holds. 
% Otherwise, $ \lceil 2D - (\lfloor D \rfloor + 1) \rceil = \lceil D \rceil$ holds. 
\end{proposition}
\begin{proof}
 Notice that 
\begin{eqnarray*}
 \lfloor 2D - (\lfloor D \rfloor + 1) \rfloor 
 = \lfloor 2(D - \lfloor D \rfloor) + \lfloor D \rfloor  -1  \rfloor
 = \lfloor 2(D - \lfloor D \rfloor)\rfloor + \lfloor D \rfloor  -1  
%\label{eq:epsilon-omit1}
\end{eqnarray*}
 holds. 
 If $D - \lfloor D \rfloor < 1/2$ then $\lfloor 2(D - \lfloor D \rfloor) \rfloor = 0$,  
 otherwise $\lfloor 2(D - \lfloor D \rfloor) \rfloor = 1$, and we obtain the claim. 
\end{proof}

%%%%%%%%%%%%%%%%%%%%%%%%%
\begin{lemma}\label{lem:epsilon-omit-midpoint}
Suppose $D>0$ is not an integer.  
Let $d = 2D - (\lfloor D \rfloor + 1)$, then 
\begin{eqnarray*}
 D = 
\begin{cases}
 \lfloor D \rfloor + \dfrac{d - \lfloor d \rfloor}{2} 
   &  (\mbox{if $D - \lfloor D \rfloor < \frac{1}{2}$}) \\
 \lfloor D \rfloor + \dfrac{d - \lfloor d \rfloor}{2} + \dfrac{1}{2}
   &  (\mbox{otherwise}) 
\end{cases}
\end{eqnarray*}
 holds. 
\end{lemma}
%%%%%%%%%%%%%%%%%%%%%%%%%%%%%%%%%%%%%%%%%%%%%%
\begin{proof}
%%%%%%%%%%%%%
 By Proposition~\ref{obs:epsilon-omit1}, 
\begin{eqnarray}
 \lfloor d \rfloor = \lfloor 2D - (\lfloor D \rfloor + 1) \rfloor = 
\begin{cases}
 \lfloor D \rfloor -1 &  (\mbox{if $D - \lfloor D \rfloor < \frac{1}{2}$)} \\
 \lfloor D \rfloor    &  (\mbox{otherwise})
\end{cases}
\label{eq:epsilon-omit2}
\end{eqnarray}
 hold. 
 Since $2D = d + (\lfloor D \rfloor + 1)$ holds by the definition of $d$, 
\begin{eqnarray*}
 D = \frac{d + (\lfloor D \rfloor + 1)}{2} 
   = \lfloor D \rfloor + \frac{d - \lfloor D \rfloor + 1}{2}
   = 
\begin{cases}
 \lfloor D \rfloor + \dfrac{d - \lfloor d \rfloor}{2} &  (\mbox{if $D - \lfloor D \rfloor < \frac{1}{2}$}) \\
 \lfloor D \rfloor + \dfrac{d - \lfloor d \rfloor + 1}{2} &  (\mbox{otherwise}) 
\end{cases}
\end{eqnarray*}
 where the last equality follows \eqref{eq:epsilon-omit2}. 
 Now, we obtain the claim.
\end{proof}

 Finally, we remark an easy implication of Proposition~\ref{obs:epsilon-omit1}. 
%%%%%%%%%%%%%%%%%%%%%%%%%%
\begin{lemma}\label{lem:epsilon-omit-parity}
 Let $d = 2D - (\lfloor D \rfloor + 1)$. 
 Then, $D - \lfloor D \rfloor \geq 1/2$ if and only if
\begin{eqnarray*}
 \lfloor D \rfloor \equiv \lfloor d \rfloor \pmod{2}.
%\label{eq:parity}
\end{eqnarray*}
\end{lemma}
\begin{proof}
By Proposition~\ref{obs:epsilon-omit1}. 
\end{proof}
 Thus, we can calculate $D$ 
  from $\lfloor D \rfloor$, $d$ and 
  the parities of $\lfloor D \rfloor$ and $\lfloor d \rfloor$. 
 We will prove Theorem~\ref{thm:epsilon-omit} in a constructive way based on this idea. 

\begin{figure}[tp]
\begin{center}
\includegraphics[width=140mm]{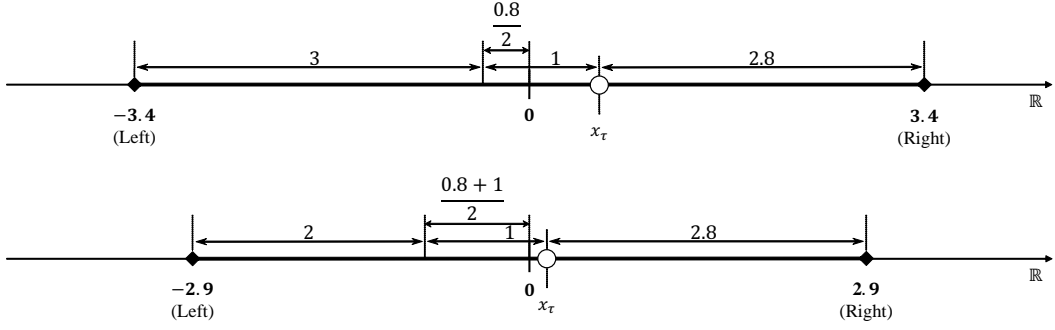}
\caption{
 Another example of the basic idea to solve Problem~\ref{prob:epsilon-omit}. 
 Similar to  Figure~\ref{fig:3.1.1a}, suppose $d(\tau)=2.8$. 
 Then there are two possibilities that $D=3.4$ or $2.9$. 
}
\label{fig:3.1.1b}
\end{center}
\end{figure}

%%%%%%%%%%%%%%%%%%%%%%%%
\subsubsection{Proof of Theorem~\ref{thm:epsilon-omit}}
%%%%%%%%%%%%%%%%%%
 For convenience, we define 
\begin{eqnarray}
 H^-_{\epsilon/4} &=& \left\{ u \in \mathbb{R}_{>0} \mid u - \lceil  u \rceil  \in \left(-\tfrac{\epsilon}{4},0 \right)\right\} \\
 H^+_{\epsilon/4} &=& \left\{ u \in \mathbb{R}_{>0} \mid u - \lfloor u \rfloor \in \left(0, \tfrac{\epsilon}{4} \right)\right\}
\end{eqnarray}
 and we also define $H^{\pm}_{\epsilon/4} = H^-_{\epsilon/4} \cup H^+_{\epsilon/4}$. 
 For any $u \in H^{\pm}_{\epsilon/4}$, let 
\begin{eqnarray}
 [u] = \begin{cases}
 \lceil u \rceil   & \mbox{(if $u \in H^-_{\epsilon/4}$)} \\
 \lfloor u \rfloor & \mbox{(if $u \in H^+_{\epsilon/4}$)} 
 \end{cases}
\end{eqnarray}
 and let 
\begin{eqnarray}
 \Delta(u) = |u-[u]| 
\end{eqnarray}
 for any $u \in H^{\pm}_{\epsilon/4}$. 

%%%%%%%%%%
 Now, we prove Theorem~\ref{thm:epsilon-omit}. 
\begin{proof}[Proof of Theorem~\ref{thm:epsilon-omit}]
 The proof is constructive. 
 We define a transition map $f\colon \mathcal{O} \to [-D,D]$ to solve Problem~\ref{prob:epsilon-omit} by 
\begin{eqnarray*}
f((L,d))&=&
\begin{cases}
 x+1 
   & \mbox{if $d \in \mathbb{Z}$} \\
 x + \Delta(d) +\frac{2}{\epsilon}\Delta(d) +\frac{1}{2} 
%   & \mbox{if $\Delta(d) \in (-\tfrac{\epsilon}{4}+\tfrac{\epsilon^2}{2},0)$ and $[d]$ is even} \\
   & \mbox{if $d \in H_{\epsilon/4}^-$, $[d]$ is even and 
     $\tfrac{2}{\epsilon} \Delta(d) \leq \tfrac{1}{2}-\epsilon$} \\
 x + \Delta(d) +\frac{2}{\epsilon}\Delta(d) 
   & \mbox{if $d \in H_{\epsilon/4}^-$, $[d]$ is odd and 
     $\tfrac{2}{\epsilon} \Delta(d) \geq \epsilon$} \\
 x - \Delta(d) +\frac{2}{\epsilon}\Delta(d) 
   & \mbox{if $d \in H_{\epsilon/4}^+$, $[d]$ is even and 
     $\tfrac{2}{\epsilon} \Delta(d) \geq \epsilon$} \\
 x - \Delta(d) +\frac{2}{\epsilon}\Delta(d) +\frac{1}{2} 
   & \mbox{if $d \in H_{\epsilon/4}^+$, $[d]$ is odd and 
     $\tfrac{2}{\epsilon} \Delta(d) \leq \tfrac{1}{2}-\epsilon$} \\
 x - d + \lfloor d \rfloor & \mbox{otherwise}\\
\end{cases} \\
f((R,d))&=& 
\begin{cases}
 x - \frac{1}{2} 
  & \mbox{if $d \in \mathbb{Z}$}\\
 x - 1 - \frac{\epsilon}{2} \cdotp \frac{d-\lfloor d \rfloor}{2} 
  & \mbox{if $d \not\in \mathbb{Z}$ and $\lfloor d \rfloor$ is even}\\
 x - 1 + \frac{\epsilon}{2} \cdotp \frac{d-\lfloor d \rfloor}{2}
  & \mbox{if $d \not\in \mathbb{Z}$ and $\lfloor d \rfloor$ is odd}
\end{cases}\\
f(\Myquit) &=& x
 \end{eqnarray*}
   in each case of $\phi(D,x) = (L,d)$, $(R,d)$ or $\Myquit$
  for any $x \in [-D,D]$ (see also Algorithm~\ref{alg:epsilon-omit}).
%%%%%
 It is not difficult to observe that 
   $f$ is a transition map (recall Problem~\ref{prob:original}). 
 For convenience, 
  let $(\SIDE(t),d(t))=\phi(D,x_t)$. 
%%%%%
\begin{figure}[tp]
\begin{center}
\includegraphics[width=140mm]{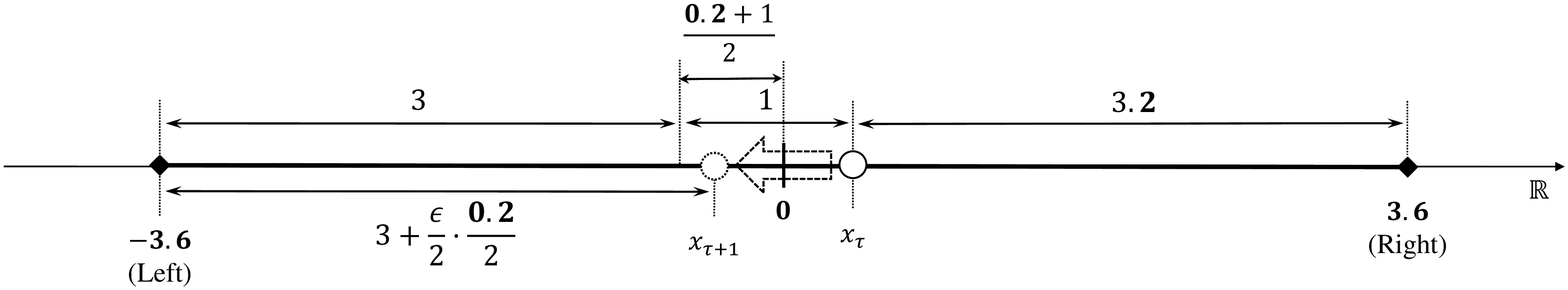} \\
\includegraphics[width=140mm]{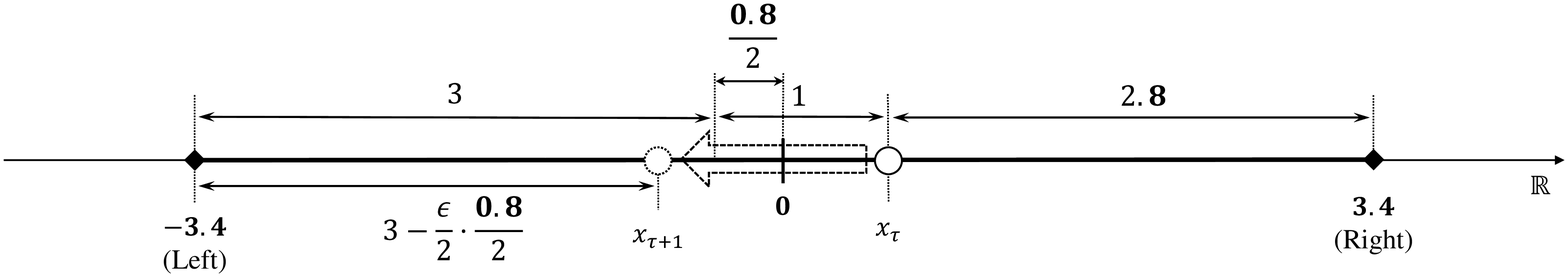}
\caption{
 Example of the motion by $f(R,d(\tau))$. 
 If $\lfloor d(\tau) \rfloor = 3$, thus $\lfloor d(\tau) \rfloor$ is odd, 
  then $x_{\tau+1}$ is slightly right to the point $-D+\lfloor D \rfloor$, 
   i.e., $d(\tau+1)>\lfloor D \rfloor=3$ (upper figure). 
 If $\lfloor d(\tau) \rfloor = 2$, thus $\lfloor d(\tau) \rfloor$ is even, 
  then $x_{\tau+1}$ is slightly left to the point $-D+\lfloor D \rfloor$, 
   i.e., $d(\tau+1)<\lfloor D \rfloor=3$ (lower figure). 
}\label{fig:3.2.2a}
\end{center}
\end{figure}

%%%% 
 To begin with, we claim that we may assume that $x_0 \in [-D,0)$. 
 By the definition of $f$, 
  Walker moves to left with distance at least $\frac{1}{2}$ whenever he observes the right-end. 
 Thus, Walker eventually moves into the segment $[-D,0]$. 
 It is trivial if $x_0 = 0$.

%%%%
 In the following, we are concerned with the case that $x_0 \in [-D,0)$. 
 Firstly, we are concerned with the case that $d(0) \in \mathbb{Z}$. 
 Then Walker moves to right with distance one by a step, and 
  he eventually observes the right-end at time $\tau$ where $x_{\tau} = -D + \lfloor D \rfloor + 1$\footnote{
   Here we remark that $D$ is not an integer 
    since $D \in \overline{\mathbf{S}}_{\epsilon}$. 
   This is not essential. 
  }. 
 It is not difficult to observe that $d(\tau) = 2D-\lfloor D \rfloor + 1$. 
 Thus, Lemmas~\ref{lem:epsilon-omit-midpoint} and~\ref{lem:epsilon-omit-parity} imply that 
\begin{align}
 &x_{\tau} - 1 + \frac{d(\tau)-\lfloor d(\tau) \rfloor}{2} = 0 
 && (\mbox{if $\lfloor D \rfloor \not\equiv \lfloor d(\tau) \rfloor \pmod{2}$})
 \label{eq:epsilon-omit-tau1}\\
 &x_{\tau} - 1 + \frac{d(\tau)-\lfloor d(\tau) \rfloor}{2} +\frac{1}{2} = 0 
 && (\mbox{if $\lfloor D \rfloor \equiv \lfloor d(\tau) \rfloor \pmod{2}$})
% && (\mbox{otherwise, i.e., $\lfloor D \rfloor \equiv \lfloor d(\tau) \rfloor \pmod{2}$})
\label{eq:epsilon-omit-tau2}
\end{align}
 where $0$ is the goal of the problem. 
%  $x_{\tau} - 1 + \frac{d(\tau)-\lfloor \rfloor}{2} = 0$ if $D - \lfloor D \rfloor<1/2$, and 
%  $x_{\tau} - 1 + \frac{d(\tau)-\lfloor \rfloor}{2} = 0$ if $\lfloor D \rfloor$, otherwise 
 If $d(\tau) \in \mathbb{Z}$ then we observe from \eqref{eq:epsilon-omit-tau2} that 
   $x_{\tau+1} = x_{\tau} -\tfrac{1}{2} = 0$ 
   according to the function $f$. 
 Otherwise, 
  Walker moves to 
\begin{eqnarray}
  x_{\tau+1} 
  = \begin{cases}
  -D + \lfloor D \rfloor - \dfrac{\epsilon}{2} \cdotp \dfrac{d(\tau)-\lfloor d(\tau) \rfloor}{2} 
  &(\mbox{if $\lfloor d(\tau) \rfloor$ is even}) \\
  -D + \lfloor D \rfloor + \dfrac{\epsilon}{2} \cdotp \dfrac{d(\tau)-\lfloor d(\tau) \rfloor}{2}
  &(\mbox{if $\lfloor d(\tau) \rfloor$ is odd}) \\
\end{cases}
\end{eqnarray}
   according to $f$ (see Figure~\ref{fig:3.2.2a}). 
 It is easy to see that $x_{\tau+1} < 0$ holds when $\lfloor d(\tau) \rfloor$ is even, 
  and also $x_{\tau+1} < 0$ holds when $\lfloor d(\tau) \rfloor$ is odd  
  since 
  \eqref{eq:epsilon-omit-tau1} and \eqref{eq:epsilon-omit-tau2}
  with the fact that 
   $\frac{\epsilon}{2} \cdotp \frac{d(\tau)-\lfloor d(\tau) \rfloor}{2} < \frac{d(\tau)-\lfloor d(\tau) \rfloor}{2}$. 
 It is also easy to see that $d({\tau+1}) \in H_{\epsilon/4}^{\pm}$, and 
  then Walker moves to $x_{\tau+2} = 0$ according to $f$ by \eqref{eq:epsilon-omit-tau1} and \eqref{eq:epsilon-omit-tau2}. 
 Thus, we obtain the desired situation in this case. 
  
%%%%%
 If $d(0) \not\in \mathbb{Z}$ and $d(0) \not\in H_{\epsilon/4}^{\pm}$, then 
  $d(1) \in \mathbb{Z}$ and the case is easily reduced to the above.

%%%%%
 In the rest of the proof, 
   we are concerned with the remaining case, that is $d(0) \in H_{\epsilon/4}^{\pm}$. 
 For convenience, 
  let $x^* = -d(0) + [d(0)]$, 
  i.e., 
  $x^* = x_0 + \Delta(d(0))$ if $d(0) \in H_{\epsilon/4}^-$, and 
  $x^* = x_0 - \Delta(d(0))$ if $d(0) \in H_{\epsilon/4}^+$. 
 Firstly, 
  we observe from the definition of $f$ that 
\begin{eqnarray}
  x^* + \epsilon \leq x_1 \leq x^* + 1-\epsilon
\label{tmp180921a}
\end{eqnarray}
  hold by the definition of $f$. 
 If $\SIDE(1)=L$ then 
 \eqref{tmp180921a} implies that $d(1) \not \in H_{\epsilon/4}^{\pm}$, and hence 
   the case is reduced to one of the cases above discussed. 
 Suppose $\SIDE(1)=R$. 
%%%%%
 Then $x_2$ satisfies 
\begin{eqnarray}
 x_1-1-\frac{\epsilon}{4}  \leq x_2 \leq x_1-1+\frac{\epsilon}{4} 
\label{tmp180921b}
\end{eqnarray} 
 according to $f$. 
 Combining \eqref{tmp180921a} with \eqref{tmp180921b}, we obtain 
\begin{eqnarray}
 x^*-1+\frac{3}{4}\epsilon \leq x_2 \leq x^*-\frac{3}{4}\epsilon 
\label{tmp180921c}
\end{eqnarray} 
  which implies $\SIDE(2)=L$. 
 Besides, 
  \eqref{tmp180921c} also implies that 
   $\frac{3}{4} \epsilon \leq |d(2) - \lfloor d(2) \rfloor| \leq 1-\frac{3}{4} \epsilon$ hold, 
%  both 
%   $|d(2) - \lfloor d(2) \rfloor| \geq \frac{3}{4} \epsilon$ and 
%   $|d(2) - \lceil  d(2) \rceil | \geq \frac{3}{4} \epsilon$ hold, 
  meaning that 
   $d(2) \not\in H_{\epsilon/4}^{\pm}$. 
 Thus, the case is also reduced to the above discusses case. 
 We obtain the claim. 
\end{proof}

%%%%%%%%%%%%%%%
\begin{algorithm}
\caption{($D - [D] \not\in (-\epsilon,\epsilon)$)}
\label{alg:epsilon-omit}
\algsetup{indent=1.5em}
\begin{algorithmic}[1]
\LOOP
\STATE observe $(\SIDE,d)$ or $\Myquit$
	\IF {$\SIDE=L$}
		\IF {$d \in \mathbb{Z}$}
			\STATE move right by $1$ 
		\ELSIF{[$d - \lceil d \rceil \in (-\frac{\epsilon}{4}+\frac{\epsilon^2}{2},0)$] and 
			    [$\lceil d \rceil$ is even]}
				\STATE move right by 
				 $\left(1+ \frac{2}{\epsilon}\right)(\lceil d \rceil - d) + \frac{1}{2}$ 
		\ELSIF{[$d - \lceil d \rceil \in (-\frac{\epsilon}{4},-\frac{\epsilon^2}{2})]$ and 
			    [$\lceil d \rceil$ is odd]}
				\STATE move right by 
				 $\left(1+ \frac{2}{\epsilon}\right)(\lceil d \rceil - d)$ 
		\ELSIF{[$d - \lfloor d \rfloor \in (\frac{\epsilon^2}{2},\frac{\epsilon}{4})$] and 
			    [$\lceil d \rceil$ is even]}
				\STATE move right by 
				 $\left(-1+ \frac{2}{\epsilon}\right)(d-\lfloor d \rfloor)$
		\ELSIF{[$d - \lfloor d \rfloor \in (0,\frac{\epsilon}{4}-\frac{\epsilon^2}{2})$] and 
			    [$\lceil d \rceil$ is even]}
				\STATE move right by 
				 $\left(-1+ \frac{2}{\epsilon}\right)(d-\lfloor d \rfloor) + \frac{1}{2}$ 
		\ELSE
				\STATE move left by $d- \lfloor d \rfloor$ 
		\ENDIF
	\ELSIF {$\SIDE=R$}
		\IF {$d \in \mathbb{Z}$}
	  	  	  \STATE move left by $\frac{1}{2}$
		\ELSIF {$\lfloor d \rfloor$ is even}
	  	  	  \STATE move left by $1 - \frac{(d-\lfloor d \rfloor)\epsilon}{4}$
		\ELSE
		  \STATE move left by $1 + \frac{(d-\lfloor d \rfloor)\epsilon}{4}$
		\ENDIF
	\ELSE
	  \STATE (i.e., $\Myquit$ is observed) stay there
	\ENDIF
\ENDLOOP
\end{algorithmic}
\end{algorithm}

%%%%%%%%%%%%%%%%%%%%%%%%%%%%%%%%%%%%%%%%%%%%%%%%
\subsection{$D$ chosen from almost everywhere in $\mathbb{R}_{\geq 1}$}\label{sec:nonCantor}
 Section~\ref{sec:nonCantor} briefly mentions to another interesting example, 
   where $D$ ($D \geq 1$) is chosen from almost everywhere. 
% We in Section~\ref{sec:epsilon-omit} presented an example 
%  that the restricted $D$ is from an uncountable set. 
% Furthermore, the set has positive measure. 
% This fact is enhanced by the following Theorem~\ref{thm:nonCantor}, 
%  where we are concerned with $D$ almost everywhere in $\mathbb{R}_{\geq 1}$. 
%%%%%%%%%%%%%%%%%%
%\subsubsection{Problem}\label{sec:prob-nonCantor}
 The {\em Cantor set} $\mathbf{T} \subset [0,1]$ is given by 
\begin{eqnarray}
 \mathbf{T} = \left\{ x \in \mathbb{R} \ \middle|\ x = \sum_{i=1}^{\infty} d_i 3^{-i} \mbox{ where $d_i \in \{0,2\}$ for $i=1,2,\ldots$} \right\}. 
\end{eqnarray}
%%%%%%%%%%%
 Extending $\mathbf{T}$ to reals other than $[0,1]$, we define 
\begin{eqnarray}
 \mathbf{T}_{\rm ex} = \left\{ x \in \mathbb{R} \ \middle|\ x - \lfloor x \rfloor \in \mathbf{T} \right\}. 
\end{eqnarray}
 In this paper, we say $r \in \mathbb{R}$ is a {\em Cantor real}\/\footnote{
 We specially remark that 
  the condition is different from 
   $r \in \left\{ x \in \mathbb{R} \ \middle|\ x = \sum_{i=-\infty}^{\infty} d_i 3^{-i} 
  \mbox{ where $d_i \in \{0,2\}$ for $i\in \mathbb{Z}$}\right\}$, which is also called Cantor real in some context. 
 } 
 if $r \in \mathbf{T}_{\rm ex}$. 
  We are concerned with the following problem. 
%%%%%%%%%%%%
\begin{problem}[$D$ is {\em not} a Cantor real]\label{prob:nonCantor}
 As given 
  the observation function $\phi \colon \mathbb{R} \times [-D,D] \to \mathcal{O}$ defined by \eqref{def:phi}, 
%%%%%%%%
 the goal of the problem is to design a transition map 
  $f \colon \mathcal{O} \to [-D,D]$ 
  for which the potential function $\Psi(D,x)$, defined by \eqref{def:Psi}, is bounded 
  for any $D$ ($1 < D < \infty$) such that \underline{$D \not\in \mathbf{T}_{\rm ex}$} and any real $x \in [-D,D]$. 
\end{problem}
%%%%%%%%%%%%%%%%
% Notice that the cardinality of $\mathbb{R} \setminus \mathbf{T}_{\rm ex}$ is equal to that of $\mathbb{R}$. 
% Furthermore, almost all reals are not Cantor reals; 
 It is known that the Lebesgue measure of the Cantor set is zero, 
  which implies that the Lebesgue measure of $[0,1] \setminus \mathbf{T}$ is equal to $1$. 
% In this paper, we call $x \in \mathbb{R}$ a {\em Cantor real} if $x - \lfloor x \rfloor \in \mathbb{T}$. 
% It is known that the cardinality of $\mathbf{T}$ is equal to that of $\mathbb{R}$ (see e.g., \cite{?}).
% Furthermore, 
%  the Luberg measure of $(\mathbb{R} \setminus \mathbf{T}) \cap [1,r] = 1$ is for any $r \gg 1$. 
%%%%%%%%%%%%%%%%%%%%%%%%%%%%%%%%%%%%%%%%%%
\begin{theorem}\label{thm:nonCantor}
 Problem~\ref{prob:nonCantor} is solvable. 
\end{theorem}
 The proof of Theorem~\ref{thm:nonCantor} is similar to that of Theorem~\ref{thm:epsilon-omit}, 
  but it needs to be more carefully crafted. 
 See Section~\ref{apx:nonCantor}.

\subsection{Remarks on Section~\ref{sec:restriction}}
%%%%%%%%%%%%%%%%%%%%%%
 Let $\mathbf{S}_{\epsilon} = \{ x \in \mathbb{R} \mid x-[x] \in (-\epsilon,\epsilon) \}$ for an arbitrary small $\epsilon$ ($0<\epsilon<1/2$) where $[x]$ for any real $x$ denotes an integer minimizing $|x - [x]|$. 
 We have shown the solvability 
  when $D$ is NOT in $\mathbf{S}_{\epsilon}$ in Section~\ref{sec:epsilon-omit}, and 
  when $D$ is NOT in $\mathbf{T}_{\rm ex}$ in Section~\ref{sec:nonCantor}. 
 We remark that their compliment cases, 
   precisely Problems~\ref{prob:epsilon} and \ref{prob:Cantor} described below, are also solvable. 
%%%%%%%%%%%
\begin{problem}[$D$ in $\mathbf{S}_{\epsilon}$]\label{prob:epsilon}
 As given 
  the observation function $\phi \colon \mathbb{R} \times [-D,D] \to \mathcal{O}$ defined by \eqref{def:phi}, 
%%%%%%%%
 the goal of the problem is to design a transition map 
  $f \colon \mathcal{O} \to [-D,D]$ 
  for which the potential function $\Psi(D,x)$, defined by \eqref{def:Psi}, is bounded 
  for any $D$ ($1 < D < \infty$) such that \underline{$D$ in $\mathbf{S}_{\epsilon}$} and any real $x \in [-D,D]$. 
\end{problem}
%%%%%%%%%%%%%%%%%%%%%%%%%%%%%%%%%%%%%%%%%%
%%%%%%%%%%%%
\begin{problem}[$D$ is a Cantor real]\label{prob:Cantor}
 As given 
  the observation function $\phi \colon \mathbb{R} \times [-D,D] \to \mathcal{O}$ defined by \eqref{def:phi}, 
%%%%%%%%
 the goal of the problem is to design a transition map 
  $f \colon \mathcal{O} \to [-D,D]$ 
  for which the potential function $\Psi(D,x)$, defined by \eqref{def:Psi}, is bounded 
  for any $D$ ($1 < D < \infty$) such that \underline{$D \in \mathbf{T}_{\rm ex}$} and any real $x \in [-D,D]$. 
\end{problem}
%%%%%%%%%%%%%%%%

 The location problem might be always solvable 
  if an uncountable set is excluded from possible $D$. 
 More precisely, the solvability of the following general problem is unsettled, 
   where we conjecture it is essentially true. 
%%%%%%%%%%%%
\begin{problem}[$D$ is not in an uncountable set]\label{prob:uncountable}
 Let $S \subseteq \mathbb{R}_{\geq 1}$ be an uncountable set. 
 As given 
  the observation function $\phi \colon \mathbb{R} \times [-D,D] \to \mathcal{O}$ defined by \eqref{def:phi}, 
%%%%%%%%
 the goal of the problem is to design a transition map 
  $f \colon \mathcal{O} \to [-D,D]$ 
  for which the potential function $\Psi(D,x)$, defined by \eqref{def:Psi}, is bounded 
  for any $D$ ($1 < D < \infty$) such that \underline{$D \not\in \mathbf{S}$} and any real $x \in [-D,D]$. 
\end{problem}
%%%%%%%%%%%%%%%%

 In contrast, 
   the complement of Problem~\ref{prob:algebraic} seems not solvable. 
 Even for the following easier version, 
  the solvability is unsettled. 
\begin{problem}[Irrational $D$]\label{prob:irrational}
 As given 
  the observation function $\phi \colon \mathbb{R} \times [-D,D] \to \mathcal{O}$ defined by \eqref{def:phi}, 
%%%%%%%%
 the goal of the problem is to design a transition map 
  $f \colon \mathcal{O} \to [-D,D]$ 
  for which the potential function $\Psi(D,x)$, defined by \eqref{def:Psi}, is bounded 
  for any \underline{irrational} $D$ ($1 < D < \infty$) and any real $x \in [-D,D]$. 
\end{problem}
%%%%%%%%%%%%%%%%%%%%%%%%%%%%%%%%%%%%%%%%%%

%%%%%%%%%%%%%%%%%%%%%%%%%%%%%%%%%%%%%%%%%%%%%%%%%%%%%
\section{Relaxation 3: With a Single-bit Memory}\label{sec:memory}
%%%%%%%%%%%
 Memoryless is definitely a property which makes the problem difficult  
  because Problem~\ref{prob:original} is easily solved if Walker has enough memory (recall Section~\ref{sec:intro}). 
%%%%%%%%%%%%
 Interestingly, this section shows that 
  only a single-bit memory is sufficient for a {\em self-stabilizing} localization of the midpoint. 
 The problem, with which this section is concerned, is formally described as follows. 
%%%%%%%%%%%%%%%%%%%%%%
\begin{problem}[With a single-bit memory]\label{prob:memory}
 As given 
  the observation function $\phi \colon \mathbb{R} \times [-D,D] \to \mathcal{O}$ defined by \eqref{def:phi}, 
 the goal of the problem is to design a transition map \underline{with memory} 
  $f \colon \mathcal{O} \times \{0,1\} \to [-D,D] \times \{0,1\}$ 
  for which an integer $n$ ($0 \leq n <\infty$) exists 
 for any real $D$ ($1<D<\infty$), real $x_0 \in [-D,D]$ and $b_0 \in \{0,1\}$
  such that $x_n=0$ 
   where $(x_{i+1},b_{i+1}) = f(\phi(D,x_i),b_i)$ for $i=0,1,2,\ldots$. 
\end{problem}
%%%%%%%%%%%%%%%%%%%%%%%%%%%%%%%%%%
\begin{theorem}\label{thm:memory}
Problem~\ref{prob:memory} is solvable. 
\end{theorem}
\begin{proof}
 The proof is constructive. 
 We define a transition map $f\colon\mathcal{O} \times \{0,1\} \to [-D,D] \times\{0,1\}$ 
  to solve Problem~\ref{prob:memory} by 
\begin{eqnarray*}
f((L,d),b) &=&\begin{cases}
(x+1,(d+1)\bmod 2)
& \mbox{if $d \in \mathbb{Z}_{\geq 0}$} \\
(x-d+\lfloor d \rfloor, \lfloor d \rfloor\bmod{2}) 
& \mbox{if $d \not\in \mathbb{Z}_{\geq 0}$}\\
\end{cases}\\ 
 f((R,d),b)) &=&\begin{cases}
\left(x-1+\tfrac{d-\lfloor d \rfloor}{2},(b+1) \bmod 2 \right) 
& \mbox{if $b \not\equiv \lfloor d \rfloor \pmod 2$}\\[2ex]
\left(x-1+\tfrac{d-\lfloor d \rfloor+1}{2},(b+1) \bmod 2\right) 
& \mbox{if $b \equiv \lfloor d \rfloor \pmod 2$}\\
\end{cases}\\
 f(\Myquit,b) &=& (x,b)
\end{eqnarray*}
   in each case of $\phi(D,x) = (L,d)$, $(R,d)$ or $\Myquit$
  for any $x \in [-D,D]$ (see also Algorithm~\ref{alg:memory}). 
 It is not difficult to observe that 
   $f$ is a transition map (recall Problem~\ref{prob:original}). 
%   especially considering that 
%     $D = x+d$ when $\phi(D,x)=(R,d)$. 

%%%%%%%%%%%%%%%%%%
 First, we show for any $x_0 \in [-D,0)$ that 
  a finite $n \in \mathbb{Z}_{> 0}$ exists such that $x_n=0$ 
  where $(x_t,b_t) = f(\phi(D,x_{t-1}),b_{t-1})$ for $t=1,2,\ldots$. 
 For convenience, 
  let $(\SIDE(t),d(t))=\phi(D,x_t)$. 
 Let $\tau = \min\{t \in \mathbb{Z}_{\geq 0} \mid \SIDE(t)=R \}$. 
%%%%%%%%
 Then, we observe that $x_{\tau} = -D + \lfloor D \rfloor + 1$, 
 and hence $d(\tau)=D-x_{\tau} = D -(-D+\lfloor D\rfloor + 1) = 2D -\lfloor D\rfloor - 1$. 
%Thus, 
% $\lfloor d(\tau) \rfloor 
%   = \lfloor 2D -\lfloor D\rfloor - 1 \rfloor 
%   = \lfloor 2D \rfloor -\lfloor D\rfloor - 1$, meaning that  
% $\lfloor d(\tau) \rfloor + \lfloor D\rfloor + 1 =  \lfloor 2D \rfloor $. 
% It is not difficult from the property of the floor function, 
%   $ 2 \lfloor D \rfloor \leq \lfloor 2D \rfloor \leq 2 \lfloor D \rfloor + 1$. 
Note that $b_{\tau} \equiv \lfloor D \rfloor + 1 \pmod{2}$ holds at that time. 
% then 
%\begin{eqnarray*}
%\lfloor d(\tau) \rfloor = \begin{cases}
% \lfloor D \rfloor & \mbox{if $b \equiv\lfloor d(\tau) \rfloor  \pmod{2}$}\\
% \lfloor D \rfloor -1 & \mbox{if $b \not\equiv\lfloor d(\tau) \rfloor  \pmod{2}$}\\
%\end{cases}
%\end{eqnarray*}
%holds. 
%Since $d(\tau) = 2D - \lfloor D \rfloor -1$, 
%\begin{eqnarray*}
%D = \begin{cases}
% \lfloor D \rfloor + \dfrac{d(\tau)  - \lfloor d(\tau) \rfloor}{2} 
%  & \mbox{if $b \equiv\lfloor d(\tau) \rfloor \pmod{2}$}\\[2ex]
% \lfloor D \rfloor + \dfrac{d(\tau)  - \lfloor d(\tau) \rfloor+1}{2} 
%  & \mbox{if $b \not\equiv\lfloor d(\tau) \rfloor  \pmod{2}$}\\
%\end{cases}
%\end{eqnarray*}
%holds. 
 Now it is not difficult 
   from Lemmas~\ref{lem:epsilon-omit-midpoint} and~\ref{lem:epsilon-omit-parity} 
  to observe that $x_{\tau+1}=0$ according to $f$. 

 Next, we claim that if $\SIDE(t)=R$ then there is $t'$ ($t' > t$) such that  $\SIDE(t')=L$ or $x_{t'}=0$, 
  meaning that it is reduced to the case $x_0 \leq 0$. 
 In fact, we show that $x_{t+3} \leq x_t-\frac{1}{2}$ holds for any $t$ as long as $\SIDE(t)=\SIDE(t+1)=\SIDE(t+2)=R$, 
  and hence it implies the claim. 
 We remark that $x_{t+1} \leq x_t$ holds when  $\SIDE(t)=R$ by the definition of the transition map $f$.  
 Suppose $\SIDE(t)=\SIDE(t+1)=\SIDE(t+2)=R$. 
 In the case that $b(s) \not\equiv \lfloor d(s) \rfloor \pmod{2}$ holds for some $s \in \{t,t+1,t+2\}$, 
  then $x_{s+1} = D - \frac{d(s) +\lfloor d(s) \rfloor+1}{2} \leq D-d(s) - \frac{1}{2}= x(s) - \frac{1}{2}$, 
   and we obtain the claim in the case. 
 In the other case, i.e., $b(s) \equiv \lfloor d(s) \rfloor  \pmod{2}$ holds for  each $s \in \{t,t+1,t+2\}$. 
 Since the parities of $b(t)$, $b(t+1)$ and $b(t+2)$ alternately changes, 
   the parities of $\lfloor d(t) \rfloor$, $\lfloor d(t+1) \rfloor$ and $\lfloor d(t+2) \rfloor$ alternately changes, too.
 This implies 
   $\lfloor d(t) \rfloor \equiv \lfloor d(t+2) \rfloor \pmod{2}$ but 
   $\lfloor d(t) \rfloor \neq \lfloor d(t+2) \rfloor $. 
 Accordingly, $d(t+2) -d(t) >1 $ holds in the case. 
 we obtain the claim. 
\end{proof}

\begin{algorithm}
\caption{(with a single-bit memory)}
\label{alg:memory}
\algsetup{indent=1.5em}
\begin{algorithmic}[1]
\STATE given initial memory bit $b \in \{0,1\}$ (adversarially) arbitrarily
\LOOP
\STATE observe $(\SIDE,d)$ or $\Myquit$
	\IF {$\SIDE=L$}
		\IF {$d \in \mathbb{Z}$}
		  	\STATE move right by $1$ 
		  	\STATE set $b := d+1 \pmod 2$
		\ELSE %($d_1 \in \mathbb{Z}$)
       	\STATE move to the left-end 
		  	\STATE set $b := 0$
		\ENDIF
	\ELSIF {$\SIDE=R$}
		\IF {$b \equiv \lfloor d  \rfloor \pmod{2}$}
       	\STATE move left by $1-\frac{d - \lfloor d \rfloor}{2}$
       	\STATE set $b:=b+1 \pmod{2}$
    	\ELSE %($\STAT = \ODD$)
       	\STATE move left by $\frac{1}{2} - \frac{d - \lfloor d \rfloor}{2}$
       	\STATE set $b:=b+1 \pmod{2}$
    	\ENDIF
	\ELSE
	  \STATE (i.e., $\Myquit$ is observed) stay there
	\ENDIF
\ENDLOOP
\end{algorithmic}
\end{algorithm}

%%%%%%%%%%%%%%%%%%%%%%%%%%%%%%%%%%%%%%%%%%%%%%%%
\section{Impossibility of Symmetric Algorithms}\label{sec:symmetric}
%%%%%%%%%%
 We conjecture Problem~\ref{prob:original} is unsolvable under some appropriate axiomatic system. 
 This section gives an easy impossibility theorem for 
   Problem~\ref{prob:original} assuming a (very strong) condition. 
 We say a transition map is {\em symmetric} if 
   $f(\phi(D,-x)) = -f(\phi(D,x))$ holds for any $x \in [-D,D]$ and for any $D \in \mathbb{R}$. 

%%%%%%%%%%%%%%%
\begin{theorem}
 No symmetric algorithm solves Problem~\ref{prob:original}. 
\end{theorem}
%%%%%%%%%%%%%%%%%%%%%%%%%%%%%%%%%%%%%%%%%
\begin{proof}
 Assume for a contradiction that $f$ is a symmetric transition map which solves Problem~\ref{prob:original}. 
 Then, 
  there is $x^* \in [-D,D] \setminus \{0\}$ such that  $f(\phi(D,x^*)) = 0$, meaning that $\Psi(D,x^*)=1$. 
 Since $f$ is symmetric, $f(\phi(D,-x^*))=-f(\phi(D,x^*))=0$ holds, too. 
 Thus, We may assume $x^* > 0$ without loss of generality. 

 Here, we remark on the observation function 
   that $\phi(D-u,x-u) = \phi(D,x)$ holds for any $D$, $x$ and $u$ ($u < x$) when $x > 0$, as well as 
   that $\phi(D-u,x+u) = \phi(D,x)$ when $x<0$. 
 Since $f$ is a transition map, meaning that $f(\phi(D,x))-x$ is independent of $x$, 
\begin{align}
 f\left(\phi\left(D-\tfrac{x^*}{2},x^*-\tfrac{x^*}{2}\right)\right)-\left(x^*-\frac{x^*}{2}\right) 
  = f\left(\phi(D,x^*)\right)-\frac{x^*}{2} 
  = -\frac{x^*}{2} 
\label{tmp170710a}
\end{align}
 holds by the assumption $f(\phi(D,x^*))=0$. 
 On the other hand, 
\begin{align}
 f\left(\phi\left(D-\tfrac{x^*}{2},-\tfrac{x^*}{2}\right)\right) 
  = -f\left(\phi\left(D-\tfrac{x^*}{2},\tfrac{x^*}{2}\right)\right) 
  = \frac{x^*}{2}
\label{tmp170710b}
\end{align}
 holds since the assumption that $f$ is symmetric. 
 It is  not difficult to see that 
  \eqref{tmp170710a} and \eqref{tmp170710b} 
  imply $\Psi(D-\frac{x^*}{2},\frac{x^*}{2})=\Psi(D-\frac{x^*}{2},-\frac{x^*}{2})=\infty$. 
 Contradiction~(see Figure~\ref{Fimposs}).   
\end{proof}
%%%%%%%%%%%%%%
\begin{figure}[tb]
\begin{center}
\includegraphics[width=70mm]{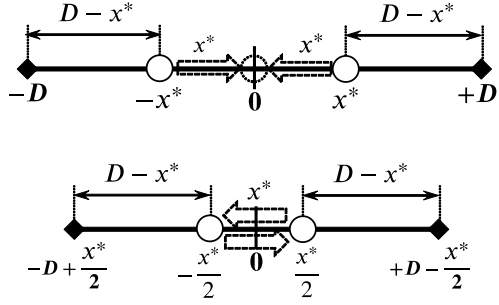}
\caption{Impossibility by a symmetric algorithm  
}
\label{Fimposs}
\end{center}
\end{figure}

%%%%%%%%%%%%%%%%%%%%%%%%%%%%%%%%%%%%%%%%%%%%%%%%%
\section{Concluding Remark}\label{sec:conclude}
%%%%%%%%%%%%%%%%%%%%%%%%%%%%
 Motivated by the theoretical difficulty of self-stabilization of 
   autonomous mobile robots with limited visibility, 
  this paper is concerned with a very simple localization problem. 
%%%%%%%%%
% The result in Section~\ref{sec:converge} implies that 
%   ``limited visibility'' makes the problem difficult.
 The techniques used in Sections~\ref{sec:converge} and~\ref{sec:restriction}
   are theoretically interesting, and 
   may indicate why the impossibility proofs of this topic are often difficult. 
 On the other hand, 
   the parity tricks used in Section~\ref{sec:memory} for a robot with a single-bit memory 
   could be reasonably simple and practically useful. 
% However, we believe that the solvability of Problem~\ref{prob:algebraic} for algebraic real numbers 
%    and Problem~\ref{prob:cantor} for a version of the Cantor set 
%    is a clue to settle Problem~\ref{prob:original}.

%%%%%%%%%%%%%%
 Problem~\ref{prob:original} remains as unsettled, and 
   we conjecture that it is unsolvable under some appropriate axiom system. 
 There are many possible variants of Problem~\ref{prob:original}. 
 A mathematically interesting version is a restriction to the rational interval, formally described as follows. 
%%%%%%%%%%%%%%%%%%%%%%
\begin{problem}[Rational domain]\label{prob:Q}
 As given 
   an observation function $\phi \colon \mathbb{Q} \times \underline{[-D,D]_{\mathbb{Q}}} \to \mathcal{O}$, 
  the goal is to design a {\em rational} transition map 
   $f \colon \mathcal{O} \to [-D,D]_{\mathbb{Q}}$ 
   such that the potential function $\Psi(D,x)$ is bounded 
  for any {\em rational} $D$ ($1 < D < \infty$), and 
  {\em rational} $x \in [-D,D]_{\mathbb{Q}}$, 
 where $[-D,D]_{\mathbb{Q}}$ denotes $[-D,D] \cap \mathbb{Q}$. 
\end{problem}
%%%%%%%%%%%%%%%%%%%%%%%%%%%%%%%%%%
\noindent
 For the version, a diagonal argument might work. 
 It is also open 
  if Problems~\ref{prob:uncountable} and \ref{prob:irrational} are respectively solvable. 

%%%%%%%%%%
 Clearly, 
  self-stabilizing coverage, spreading, pattern formation etc.\ 
  by {\em many robots with limited visibility} 
  are important future works. 

%%%%%%%%%%%%
%%%%%%%%%%%%
%%%%%%%%%%%

%%%%%%%%%%%%%%%%%%%%%%%%%%%%
%\begin{acks}
\section*{Acknowledgement}
A preliminary version appeared in \cite{mon17}\footnote{ 
  Section~\ref{sec:converge} is simplified from our preliminary version~\cite{mon17}.
  Section~\ref{sec:algebraic} is an enhancement of our preliminary version~\cite{mon17}: 
    $D$ was restricted to be rational in \cite{mon17}, while $D$ is algebraic real in this paper. 
  Sections \ref{sec:epsilon-omit} and~\ref{sec:nonCantor} are completely new. 
   }. 
This work was/is partly supported by
JSPS KAKENHI Grant Numbers JP15K15938 and JP17K19982. 
%\end{acks}

% Bibliography
%% %%%%%%%%%%%%%%%%%%%%%%%%%%%%
 
%%%%%%%%%

%\bibliographystyle{ACM-Reference-Format}
%\bibliography{TALG2017}
\appendix

%%%%%%%%%%%%%%%%%%%%%%%%%%%%%%%%%%%%%%%%%%%%%%%%
\section{$D$ almost everywhere in $\mathbb{R}_{\geq 1}$}\label{apx:nonCantor}
%%%%%%%%%%%%%%%%%%
%\subsubsection{Problem}\label{sec:prob-nonCantor}
 The {\em Cantor set} $\mathbf{T} \subset [0,1]$ is given by 
\begin{eqnarray}
 \mathbf{T} = \left\{ x \in \mathbb{R} \ \middle|\ x = \sum_{i=1}^{\infty} d_i 3^{-i} \mbox{ where $d_i \in \{0,2\}$ for $i=1,2,\ldots$} \right\}. 
\end{eqnarray}
%%%%%%%%%%%
 Extending $\mathbf{T}$ to reals other than $[0,1]$, we define 
\begin{eqnarray}
 \mathbf{T}_{\rm ex} = \left\{ x \in \mathbb{R} \ \middle|\ x - \lfloor x \rfloor \in \mathbf{T} \right\}. 
\end{eqnarray}
 In this paper, we say $r \in \mathbb{R}$ is a {\em Cantor real} if $r \in \mathbf{T}_{\rm ex}$\footnote{
 In some context, 
   $r \in \left\{ x \in \mathbb{R} \ \middle|\ x = \sum_{i=-\infty}^{\infty} d_i 3^{-i} 
  \mbox{ where $d_i \in \{0,2\}$ for $i\in \mathbb{Z}$}\right\}$  is called Cantor real. 
 } 
 We are concerned with the following problem. 
%%%%%%%%%%%%
\begin{problem}[Problem~\ref{prob:nonCantor}: $D$ is {\em not} a Cantor real]\label{apx-prob:nonCantor}
 As given 
  the observation function $\phi \colon \mathbb{R} \times [-D,D] \to \mathcal{O}$ defined by \eqref{def:phi}, 
%%%%%%%%
 the goal of the problem is to design a transition map 
  $f \colon \mathcal{O} \to [-D,D]$ 
  for which the potential function $\Psi(D,x)$, defined by \eqref{def:Psi}, is bounded 
  for any $D$ ($1 < D < \infty$) such that \underline{$D \not\in \mathbf{T}_{\rm ex}$} and any real $x \in [-D,D]$. 
\end{problem}
%%%%%%%%%%%%%%%%
 Notice that the cardinality of $\mathbb{R} \setminus \mathbf{T}_{\rm ex}$ is equal to that of $\mathbb{R}$. 
 Furthermore, almost all reals are not Cantor reals; 
  more precisely, the Lebesgue measure of $[0,1] \setminus \mathbf{T}$ is equal to $1$. 
% In this paper, we call $x \in \mathbb{R}$ a {\em Cantor real} if $x - \lfloor x \rfloor \in \mathbb{T}$. 
% It is known that the cardinality of $\mathbf{T}$ is equal to that of $\mathbb{R}$ (see e.g., \cite{?}).
% Furthermore, 
%  the Lebesgue measure of $(\mathbb{R} \setminus \mathbf{T}) \cap [1,r] = 1$ is for any $r \gg 1$. 
%%%%%%%%%%%%%%%%%%%%%%%%%%%%%%%%%%%%%%%%%%
\begin{theorem}\label{apx-thm:nonCantor}
 Problem~\ref{apx-prob:nonCantor} is solvable. 
\end{theorem}

%%%%%%%%%%%%%%%%%%%%%%%%%%%%%%%%%%%
\subsection{Proof of Theorem~\ref{apx-thm:nonCantor}}
 As a preliminary step of the Proof of Theorem~\ref{apx-thm:nonCantor}, 
  we introduce some notations. 
 For convenience, let $\epsilon = 3^{-2}$. 
 Let $x \in [0,\tfrac{1}{2}]$ be a real given by 
\begin{eqnarray}
 x = \sum_{i=1}^{\infty} \beta_i 2^{-i}
\end{eqnarray}
 with binary coefficients $\beta_i \in \{0,1\}$ $(i=1,2,\ldots)$, 
 where we employ $.01\dot{0}$ instead of $.00\dot{1}$ if $x \leq \tfrac{1}{4}$ 
 while we employ $.01\dot{1}$ instead of $.1\dot{0}$ if $x > \tfrac{1}{4}$; 
 More formally, 
  suppose for $x = \sum_{i=1}^{\infty} \beta_i 2^{-i}$ 
   that $k \in \mathbb{Z}_{>0}$ exists such that $\beta_k=1$ and $\beta_j=0 $ for any $j>k$,  
  then 
   $x = \sum_{i=1}^{\infty} \beta'_i 2^{-i}$ also holds 
  where $\beta'_i = \beta_i$ for $i<k$, $\beta'_k=0$ and $\beta'_j=1$ for $j>k$. 
 If $x \leq \tfrac{1}{4}$ we choose the coefficients $\beta_i$ ($i=1,2,\ldots$), 
 otherwise we choose $\beta'_i$ ($i=1,2,\ldots$) for $x > \tfrac{1}{4}$. 
%%%%%%%%%%%%
 We remark that $\beta_1=0$ when $x \leq 1/4$, and that $\beta'_1=0$ when $x \leq 1/2$. 
%%%%%%%%%%%%  
 Then, we define a map $h_{\rm e} \colon [0,\frac{1}{2}] \to \mathbf{T}$ by 
\begin{eqnarray}
 h_{\rm e}(x) 
  = \epsilon \sum_{i=1}^{\infty} \gamma_i 3^{-i}
  = \sum_{i=1}^{\infty} \gamma_i 3^{-i-2}
\end{eqnarray}
 where 
\begin{eqnarray*}
 \gamma_i = 
 \begin{cases}
 2 & (\mbox{if $i=2j$ and $\beta_j=1$ for $j \geq 2$}) \\
 0 & (\mbox{otherwise})
 \end{cases}
\end{eqnarray*}
 for $i=1,2,\ldots$ if $x \leq \tfrac{1}{4}$, and 
\begin{eqnarray*}
 \gamma_i = 
 \begin{cases}
 0 & (\mbox{if $i=2j$ and $\beta'_j=0$ for $j \geq 2$}) \\
 2 & (\mbox{otherwise})
 \end{cases}
\end{eqnarray*}
 for $i=1,2,\ldots$ if $x > \tfrac{1}{4}$. 
%%%%%%%
 Similarly, 
  we define a map $h_{\rm o} \colon [0,\frac{1}{2}] \to \mathbf{T}$ by 
\begin{eqnarray}
 h_{\rm o}(x) 
  = \epsilon \sum_{i=1}^{\infty} \gamma_i 3^{-i}
  = \sum_{i=1}^{\infty} \gamma_i 3^{-i-2}
\end{eqnarray}
 where 
\begin{eqnarray*}
 \gamma_i = 
 \begin{cases}
 2 & (\mbox{if $i=2j+1$ and $\beta_j=1$ for $j \geq 2$}) \\
 0 & (\mbox{otherwise})
 \end{cases}
\end{eqnarray*}
 for $i=1,2,\ldots$ if $x \leq \tfrac{1}{4}$, and 
\begin{eqnarray*}
 \gamma_i = 
 \begin{cases}
 0 & (\mbox{if $i=2j+1$ and $\beta_j=0$ for $j \geq 2$}) \\
 2 & (\mbox{otherwise})
 \end{cases}
\end{eqnarray*}
 for $i=1,2,\ldots$ if $x > \tfrac{1}{4}$. 
%%%%%%%%%%%%
 It is not difficult to observe that 
  both $h_{\rm e}$ and $h_{\rm o}$ are injective. 
 We remark that $h_{\rm e}(0)=h_{\rm o}(0)=0$, and also 
 that $h_{\rm e}(\tfrac{1}{2})=h_{\rm o}(\tfrac{1}{2})=\epsilon$. 
%%%%%%%%%%
 For convenience, we define 
\begin{eqnarray*}
 H_{\rm e}
%  &=& h_{\rm e}\left((0,\tfrac{1}{2})\right) 
   &=&  \left\{ h_{\rm e}(x) \in \mathbb{R} \mid  \ x \in (0,\tfrac{1}{2}) \right\} \\
 H_{\rm o}
%  &=& h_{\rm o}\left((0,\tfrac{1}{2})\right) 
   &=&  \left\{ h_{\rm o}(x) \in \mathbb{R} \mid  \ x \in (0,\tfrac{1}{2}) \right\}. 
\end{eqnarray*}
 Then, it is not difficult to observe that 
  $H_{\rm e} \cap H_{\rm o} = \emptyset$ holds. 
 It is not difficult to see that $h_{\rm e}$ and $h_{\rm o}$ are order preserving, 
  i.e., 
   for any $x,y \in [0,1/4]$ satisfying $x<y$, 
   both $h_{\rm e}(x) < h_{\rm e}(y)$ and $h_{\rm e}(x) < h_{\rm e}(y)$ hold. 
 The following fact is easy from the definitions of 
  $h_{\rm e}(x)$ and $h_{\rm o}(x)$. 
%%%%%%%%%%%
\begin{lemma}\label{lem:contract0}
 For any $x \in (0,\tfrac{1}{2})$, 
  $h_{\rm e}(x) \leq x$ and 
  $h_{\rm o}(x) \leq x$ hold, respectively. 
\qed
\end{lemma}

%%%%%%%%%%%
 Furthermore, 
  $h_{\rm e}$ and $h_{\rm o}$ are $\epsilon$-contractive in some case, as follows.  
\begin{lemma}\label{lem:contract1/2}
 Let $h(\tfrac{1}{2}) = \epsilon$ 
   (i.e., $h(\tfrac{1}{2}) = h_{\rm e}(\tfrac{1}{2}) = h_{\rm o}(\tfrac{1}{2})$). 
% Let $h(\tfrac{1}{2}) = h_{\rm e}(\tfrac{1}{2}) = h_{\rm o}(\tfrac{1}{2})=\epsilon$. 
 Then, 
\begin{eqnarray*}
  h(\tfrac{1}{2}) - h_{\rm e}(x) &<& \epsilon (\tfrac{1}{2}-x) \hspace{1em} \mbox{and} \\
  h(\tfrac{1}{2}) - h_{\rm o}(x) &<& \epsilon (\tfrac{1}{2}-x)
\end{eqnarray*}
  hold for any $x \in (\tfrac{1}{4},\tfrac{1}{2})$, respectively.  
\end{lemma}
\begin{proof}
 Let $x=\sum_{i=1}^{\infty} \beta'_i(x) 2^{-i}$. 
 Suppose $k \geq 2$ is the minimum index such that $\beta'_k(x) = 0$, 
  meaning that $\beta'_i(x) = 1$ for $2 \leq i < k$. 
 Then, 
  $\tfrac{1}{2} - x 
   = \sum_{i=2}^{\infty}2^{-i} - \sum_{i=2}^{\infty} \beta'_i 2^{-i} 
   = \sum_{i=2}^{\infty} (1-\beta'_i) 2^{-i} 
   = 2^{-k} + \sum_{i=k+1}^{\infty} (1-\beta'_i) 2^{-i} 
   \geq 2^{-k}$ holds. 
 Let $h_{\rm e}(x) = \epsilon\sum_{i=1}^{\infty} \gamma_i(x) 3^{-i}$. 
 Since $\beta'_i(x) = 1$ for $2 \leq i < k$, 
  $\gamma_i(x) = 2$ for $i < 2k$ 
  by the definition of $h_{\rm e}$. 
 We also remark that 
  $h(\tfrac{1}{2}) 
    = \epsilon \sum_{i=1}^{\infty} 2 \cdotp 3^{-i}$. 
 Then, 
   $h(\tfrac{1}{2}) - h_{\rm e}(x) 
    = \epsilon \sum_{i=1}^{\infty} 2 \cdotp 3^{-i} - \epsilon \sum_{i=1}^{\infty} \gamma_i(x) 3^{-i}
    = \epsilon \sum_{i=1}^{\infty} (2 - \gamma_i(x)) 3^{-i}
    = \epsilon \sum_{i=2k}^{\infty} (2 - \gamma_i(x)) 3^{-i}
    \leq \epsilon \sum_{i=2k}^{\infty} 2 \cdotp 3^{-i}
    = \epsilon 3^{-2k+1}
    < \epsilon 2^{-k} 
    \leq \epsilon (\tfrac{1}{2} - x)$. 
 Thus, we obtain the claim for $h_{\rm e}$. 
 The proof for $h_{\rm o}$ is similar. 
\end{proof}

%%%%%%%%%%%%%%%
 Now, we are ready to prove Theorem~\ref{apx-thm:nonCantor}. 
\begin{proof}[Proof of Theorem~\ref{apx-thm:nonCantor}]
 The proof is constructive. 
 For convenience, let 
\begin{eqnarray*}
 \Delta(u) = u - \lfloor u \rfloor 
\end{eqnarray*}
 for any $u \in \mathbb{R}$. 
%%%%%%%
 We define a transition map $f\colon \mathcal{O} \to [-D,D]$ to solve Problem~\ref{apx-prob:nonCantor} by 
\begin{eqnarray*}
f((L,d))&=&
\begin{cases}
 x+1 
   & \mbox{if $d \in \mathbb{Z}$} \\
 x - \Delta(d) +h_{\rm e}^{-1}(\Delta(d)) +\frac{1}{2} 
   & \mbox{if $\Delta(d) \in H_{\rm e}$, $\lfloor d \rfloor$ is even and $h_{\rm e}^{-1}(\Delta(d)) \not\in \mathbf{T}$} \\
 x - \Delta(d) +h_{\rm e}^{-1}(\Delta(d)) 
   & \mbox{if $\Delta(d) \in H_{\rm e}$, $\lfloor d \rfloor$ is odd and $h_{\rm e}^{-1}(\Delta(d)) \not\in \mathbf{T}$} \\
 x - \Delta(d) +h_{\rm o}^{-1}(\Delta(d)) 
   & \mbox{if $\Delta(d) \in H_{\rm o}$, $\lfloor d \rfloor$ is even and $h_{\rm o}^{-1}(\Delta(d)) \not\in \mathbf{T}$} \\
 x - \Delta(d) +h_{\rm o}^{-1}(\Delta(d)) +\frac{1}{2} 
   & \mbox{if $\Delta(d) \in H_{\rm o}$, $\lfloor d \rfloor$ is odd and $h_{\rm o}^{-1}(\Delta(d)) \not\in \mathbf{T}$} \\
 x - d + \lfloor d \rfloor & \mbox{otherwise}
\end{cases} \\
f((R,d))&=& 
\begin{cases}
 x - \frac{1}{2}
  & \mbox{if $d \in \mathbb{Z}$}\\
 x - 1 + h_{\rm e}(\frac{d-\lfloor d \rfloor}{2})  
  & \mbox{if $d \not\in \mathbb{Z}$ and $\lfloor d \rfloor$ is even}\\
 x - 1 + h_{\rm o}(\frac{d-\lfloor d \rfloor}{2})  
  & \mbox{if $d \not\in \mathbb{Z}$ and $\lfloor d \rfloor$ is odd}
\end{cases}\\
f(\Myquit) &=& x
 \end{eqnarray*}
   in each case of $\phi(D,x) = (L,d)$, $(R,d)$ or $\Myquit$
  for any $x \in [-D,D]$ (see also Algorithm~\ref{alg:nonCantor}).
% Here, 
%  the move is valid if both 
%  $h^{-1}(\Delta(d))\not\in H$ and  
%  $0 < h^{-1}(\Delta(d)) < \frac{1}{2}$ hold for $h \in \{h_0,h_1\}$. 
%%%%%
 It is not difficult to observe that 
   $f$ is a transition map (recall Problem~\ref{prob:original}). 
 For convenience, 
  let $(\SIDE(t),d(t))=\phi(D,x_t)$.

%%%%%%%%%%%%%%%
\begin{algorithm}
\caption{($D$ is not a Cantor real)}
\label{alg:nonCantor}
\algsetup{indent=1.5em}
\begin{algorithmic}[1]
\LOOP
\STATE observe $(\SIDE,d)$ or $\Myquit$
	\IF {$\SIDE=L$}
		\IF {$d \in \mathbb{Z}$}
		  \STATE move right by $1$
		\ELSIF{[$d - \lfloor d \rfloor \in H_{\rm e}$] $\wedge$ 
		       [$\lfloor d \rfloor$ is even]  $\wedge$ 
		       [$h_{\rm e}^{-1}(d-\lfloor d \rfloor) \not\in \mathbf{T}$]}
			\STATE move right by 
				$-(d-\lfloor d \rfloor) + h_{\rm e}^{-1}(d-\lfloor d \rfloor) +\tfrac{1}{2}$ 
		\ELSIF{[$d - \lfloor d \rfloor \in H_{\rm e}$] $\wedge$ 
		       [$\lfloor d \rfloor$ is odd]  $\wedge$ 
		       [$h_{\rm e}^{-1}(d-\lfloor d \rfloor) \not\in \mathbf{T}$]}
			\STATE move right by 
				$-(d-\lfloor d \rfloor) + h_{\rm e}^{-1}(d-\lfloor d \rfloor)$ 
		\ELSIF{[$d - \lfloor d \rfloor \in H_{\rm o}$] $\wedge$ 
		       [$\lfloor d \rfloor$ is even]  $\wedge$ 
		       [$h_{\rm o}^{-1}(d-\lfloor d \rfloor) \not\in \mathbf{T}$]}
			\STATE move right by 
				$-(d-\lfloor d \rfloor) + h_{\rm o}^{-1}(d-\lfloor d \rfloor)$ 
		\ELSIF{[$d - \lfloor d \rfloor \in H_{\rm o}$] $\wedge$ 
		       [$\lfloor d \rfloor$ is odd]  $\wedge$ 
		       [$h_{\rm o}^{-1}(d-\lfloor d \rfloor) \not\in \mathbf{T}$]}
			\STATE move right by 
				$-(d-\lfloor d \rfloor) + h_{\rm o}^{-1}(d-\lfloor d \rfloor) +\tfrac{1}{2}$ 
		\ELSE
			\STATE move left by $d- \lfloor d \rfloor$ 
		\ENDIF
	\ELSIF {$\SIDE=R$}
		\IF {$d \in \mathbb{Z}$}
			\STATE move left by $\tfrac{1}{2}$ 
		\ELSIF {$\lfloor d \rfloor$ is even}
 	  	  \STATE move left by $1 - h_{\rm e}(d-\lfloor d \rfloor)$
		\ELSE
  	  	  \STATE move left by $1 - h_{\rm o}(d-\lfloor d \rfloor)$
		\ENDIF
	\ELSE
	  \STATE (i.e., $\Myquit$ is observed) stay there
	\ENDIF
\ENDLOOP
\end{algorithmic}
\end{algorithm}

%%%% 
 To begin with, we claim that we may assume that $x_0 \in [-D,0)$. 
 By the definition of $f$, 
  Walker moves to right with distance at least $1/2$ whenever it observes the right-end, 
  where we remark that 
   $h_{\rm e}(\frac{d-\lfloor d \rfloor}{2})$ and $h_{\rm e}(\frac{d-\lfloor d \rfloor}{2})$ 
   are smaller than $\epsilon < 1/2$. 
 Thus, Walker eventually moves into the segment $[-D,0]$. 
 If $x_0 = 0$, it is trivial. 

%%%%
 In the following, we are concerned with the case that $x_0 \in [-D,0)$. 
 Firstly, we are concerned with the case that $d(0) \in \mathbb{Z}$. 
 Then Walker moves to right with distance one by a step, and 
  it eventually observes the right-end at time $\tau$ where $x_{\tau} = -D + \lfloor D \rfloor + 1$\footnote{
   Here we remark that $D$ is not an integer 
    since $D - \lfloor D \rfloor \not\in \mathbf{T}$ and $0 \in \mathbf{T}$. 
   However, this is not essential. 
  }. 
 It is not difficult to observe that $d(\tau) = 2D-\lfloor D \rfloor + 1$. 
 Thus, Lemmas~\ref{lem:epsilon-omit-midpoint} and~\ref{lem:epsilon-omit-parity} imply that 
\begin{align}
 &x_{\tau} - 1 + \frac{d(\tau)-\lfloor d(\tau) \rfloor}{2} = 0 
 && (\mbox{if $\lfloor D \rfloor \not\equiv \lfloor d(\tau) \rfloor \pmod{2}$})
 \label{eq:tau1}\\
 &x_{\tau} - 1 + \frac{d(\tau)-\lfloor d(\tau) \rfloor}{2} +\frac{1}{2} = 0 
 && (\mbox{if $\lfloor D \rfloor \equiv \lfloor d(\tau) \rfloor \pmod{2}$})
% && (\mbox{otherwise, i.e., $\lfloor D \rfloor \equiv \lfloor d(\tau) \rfloor \pmod{2}$})
\label{eq:tau2}
\end{align}
 where $0$ is the goal of the problem. 
%  $x_{\tau} - 1 + \frac{d(\tau)-\lfloor \rfloor}{2} = 0$ if $D - \lfloor D \rfloor<1/2$, and 
%  $x_{\tau} - 1 + \frac{d(\tau)-\lfloor \rfloor}{2} = 0$ if $\lfloor D \rfloor$, otherwise 
 If $d(\tau) \in \mathbb{Z}$ then $x_{\tau+1} = x_{\tau} -\tfrac{1}{2} = 0$ 
   according to the function $f$ by \eqref{eq:tau2}. 
 Otherwise, 
  Walker moves to 
   $x_{\tau+1} = -D + \lfloor D \rfloor + h(\frac{d(\tau) - \lfloor d(\tau) \rfloor}{2})$
   according to $f$, 
  where $h \in \{h_{\rm e},h_{\rm o}\}$ depends on the parity of $\lfloor d(\tau) \rfloor$.  
 At that time, 
  we remark that $x_{\tau+1} < 0$ 
   since $h(\frac{d(\tau) - \lfloor d(\tau) \rfloor}{2}) < \frac{d(\tau) - \lfloor d(\tau) \rfloor}{2}$ 
   by Lemma~\ref{lem:contract0}. 
 We also remark that $\Delta(d(\tau+1)) \in H_{\rm e} \cup H_{\rm o}$. 
 Then, 
  Walker moves to $x_{\tau+2} = 0$ according to $f$ by \eqref{eq:tau1} and \eqref{eq:tau2}. 
 Thus, we obtain the desired situation in this case. 

%%%%%
 If $d(0) \not\in \mathbb{Z}$ and $\Delta(d(0)) \not\in H_{\rm e} \cup H_{\rm o}$, then 
  $d(1) \in \mathbb{Z}$ and the case is easily reduced to the above.

%%%%%
 In the rest of the proof, 
   we are concerned with the remaining case, that is $\Delta(d(0)) \in H_{\rm e} \cup H_{\rm o}$. 
% For convenience, let $\delta = \Delta(d(0))$. 
 Firstly, we remark that if $\SIDE(1)=L$, then $\lfloor x_1 \rfloor > \lfloor x_0 \rfloor + 1$ or 
  reduced to the above cases. 
 Thus, we may assume that $x_0 > -D + \lfloor D \rfloor$, and $\SIDE(1)=R$. 
 For convenience, let $x_0 = -D + \lfloor D \rfloor + \delta$, 
   i.e., $\delta = \Delta(d(0)) \in H_{\rm e} \cup H_{\rm o}$ by the hypothesis of the case. 
 We will consider four subcases; those are 
  given by the combination of conditions 
   whether $\delta \in H_{\rm e}$ or $H_{\rm o}$ and 
   whether $\lfloor d(0) \rfloor$ is even or odd, 
  namely we will prove in the following order,  
% We will consider four subcases; those are 
 (i)   $\delta \in H_{\rm e}$ and $\lfloor d(0) \rfloor$ is odd, 
 (ii)  $\delta \in H_{\rm o}$ and $\lfloor d(0) \rfloor$ is even, 
 (iii) $\delta \in H_{\rm e}$ and $\lfloor d(0) \rfloor$ is even, 
 (iv)  $\delta \in H_{\rm o}$ and $\lfloor d(0) \rfloor$ is odd, 

%%%%%
 (i) Suppose $\delta \in H_{\rm e}$ and $\lfloor d(0) \rfloor$ is odd. 
 Then, 
  $x_1 = x_0 - \delta + h_{\rm e}^{-1}(\delta)$, and 
  $x_2 
    = x_1 - 1 + h(\frac{d(1)-\lfloor d(1) \rfloor}{2})
    = x_0 - \delta + h_{\rm e}^{-1}(\delta)- 1 + h(\frac{d(1)-\lfloor d(1) \rfloor}{2})$ 
   where $h \in \{h_{\rm e},h_{\rm o}\}$ is appropriately chosen. 
 Here we remark that $h_{\rm e}^{-1}(\delta) < \frac{1}{2}$ 
   since $\delta \in H_{\rm e}$ (recall the definition of $H_{\rm e}$), 
  thus 
  $x_2 
    < x_0 - \delta -\frac{1}{2} + h(\frac{d_1-\lfloor d_1 \rfloor}{2}) 
    = -D + \lfloor D \rfloor -\frac{1}{2} + h(\frac{d_1-\lfloor d_1 \rfloor}{2})
    < -D + \lfloor D \rfloor$ 
    where the last inequality follows $h(\frac{d_1-\lfloor d_1 \rfloor}{2}) < \frac{1}{2}$.  
 The case is reduced to the cases above discussed. 
%%%%%%%
 The case (ii) is proved in a similar way as (i). 
%  In the case that $\Delta(d(0)) \in H_{\rm o}$ and $\lfloor d(0) \rfloor$ is even, 
%  the proof is similar to (i). 

%%%%%%%%%
 Cases (iii) and (iv) are simultaneously proved. 
 If $h_{\rm e}(\delta) \leq h_{\rm e}(\frac{1}{4})$ in case of (iii) 
 (resp.,\ $h_{\rm o}(\delta) \leq h_{\rm o}(\frac{1}{4})$ in case of (iv)), 
  we obtain $x_2 < -D + \lfloor D \rfloor$ in a similar way as case (i). 
 Suppose $h_{\rm e}(\delta) > h_{\rm e}(\frac{1}{4})$ in case of (iii). 
  (resp.,\ $h_{\rm o}(\delta) > h_{\rm o}(\frac{1}{4})$ in case of (iv)). 
 For convenience, 
  let $x^* = -D + \lfloor D \rfloor + \epsilon$, 
  where we remark that $\epsilon = h_{\rm e}(\tfrac{1}{2}) = h_{\rm o}(\tfrac{1}{2})$. 
 Then, 
  we claim that 
\begin{eqnarray}
  x^* - x_2 > \epsilon^{-1}(x^* - x_0) = 9(x^* - x_0) > 0 \hspace{2em} \mbox{(recall $\epsilon=3^{-2}$, by definition)}
\label{eq:nonCantor_geometric}
\end{eqnarray}
   holds in both cases (iii) and (iv), 
   where notice that $x^* - x_0 = \epsilon-\delta > 0$. 
%%%%%%%%
 When $\delta \in H_{\rm e}$,  
  according to $f$, 
%  $x_1 = x_0 - \delta + h_{\rm e}^{-1}(\delta)+ \frac{1}{2} = x^*  + h_{\rm e}^{-1}(\delta)+ \frac{1}{2}$, and 
  $x_1 = x_0 - \delta + h_{\rm e}^{-1}(\delta)+ \frac{1}{2}$, and 
  $x_2 = x_1 - 1 + h(\frac{d(1) - \lfloor d(1) \rfloor}{2}) 
       = x_0 - \delta + h_{\rm e}^{-1}(\delta) - \frac{1}{2} + h(\frac{d(1)- \lfloor d(1) \rfloor}{2}) 
       = -D + \lfloor D \rfloor + h_{\rm e}^{-1}(\delta) - \frac{1}{2} + h(\frac{d(1)- \lfloor d(1) \rfloor}{2})$ 
%        = x^* + h_{\rm e}^{-1}(\delta) - \frac{1}{2} + h(\frac{d(1)- \lfloor d(1) \rfloor}{2})$ 
   with the appropriately $h \in \{h_{\rm e},h_{\rm o}\}$. 
 Thus, 
\begin{align*}
  x^* - x_2 
    &=  h_{\rm e}(\tfrac{1}{2}) - h_{\rm e}^{-1}(\delta) + \tfrac{1}{2} - h(\tfrac{d(1)- \lfloor d(1) \rfloor}{2})  \nonumber\\
    &> \tfrac{1}{2}-h_{\rm e}^{-1}(\delta) 
    && \left(\mbox{since $h_{\rm e}(\tfrac{1}{2}) > h(\tfrac{d(1)- \lfloor d(1) \rfloor}{2})$}\right)\nonumber\\
    &= h_{\rm e}^{-1}(h_{\rm e}(\tfrac{1}{2})) -h_{\rm e}^{-1}(\delta) \nonumber\\
    &> \tfrac{1}{\epsilon}(h_{\rm e}(\tfrac{1}{2})-\delta)
    && (\mbox{by Lemma~\ref{lem:contract1/2}}) \nonumber\\
    &= \tfrac{1}{\epsilon}(x^*-x_0),
\end{align*}
 and we obtain \eqref{eq:nonCantor_geometric} in case of $\delta \in H_{\rm e}$. 
 Similarly, \eqref{eq:nonCantor_geometric} is obtained in case of $\delta \in H_{\rm o}$. 
 These facts inductively imply that there is $t$ such that 
   $x^*-x_{2t} > 9^t(x^*-x_0) > \epsilon$, 
   meaning that we eventually obtain 
   $x_{2t} < -D + \lfloor D \rfloor$ at time $t$ ($t \leq -\log_9 (x^*-x_0) $), 
  unless we obtain $t' < t$ such that $x_{2t'}$ is a solvable points, 
   i.e., $\Delta(d(2t')) \not \in H_{\rm e} \cup H_{\rm e}$ or $h(\Delta(d(2t'))) > \tfrac{1}{4}$.  
 In any case, we obtain that all cases are reduced to a solvable case. 
\end{proof}
\end{document}